\documentclass[11pt]{article}
\usepackage{geometry} 
\usepackage[T1]{fontenc}
\usepackage[utf8]{inputenc}
\usepackage[english]{babel} 
\usepackage{amsthm} 
\usepackage{amsmath}
\usepackage{amssymb} 
\usepackage{mathtools,xparse} 
\usepackage{xargs}
\usepackage[colorinlistoftodos,prependcaption,textsize=tiny]{todonotes}
\usepackage{algorithm} 
\usepackage[noend]{algpseudocode} 
\usepackage{color}
\usepackage{tikz}
\usetikzlibrary{shapes.geometric, arrows,automata} \usepackage{balance}
\usepackage{graphicx}
\usepackage[normalem]{ulem}
\usepackage{breqn}

\newcommand{\ignore}[1]{}
\newcommand{\findmin}{\mathsf{Find \ Minimum}}
\newcommand{\findmax}{\mathsf{Find \ Maximum}}

\newcommand{\parens}[1]{\left( #1 \right)}
\newcommand{\floor}[1]{\left\lfloor #1 \right\rfloor}
\newcommand{\ceil}[1]{\left\lceil #1 \right\rceil}
\newcommand{\polylog}{\mathrm{polylog}}
\newcommand{\Prob}[1]{\mathbf{P}\left(#1\right)}

\newtheorem{theorem}{Theorem}[section]

\newtheorem{lemma}{Lemma}[section]
\newtheorem{claim}{Claim}

\newtheorem{invariant}{Invariant}[section]
\newtheorem{remark}{Remark}[section]
\usepackage[utf8]{inputenc}

\newcommand{\ball}[3]{\mathsf{Ball}_{#1}(#2, #3)}
\newcommand{\clusterGraph}[2]{{\mathsf{cluster}}(#1,#2)}
\newcommand{\cluster}{{\mathsf{Cl}}}
\newcommand{\bfs}{\mathsf{Recursive}\text{-}\mathsf{BFS}}
\newcommand{\diam}{\mathsf{diam}}
\newcommand{\dist}{\mathsf{dist}}
\newcommand{\sr}{\mathsf{Local}{\text -}\mathsf{Broadcast}}
\newcommand{\srs}{\mathsf{Local}{\text -}\mathsf{Broadcast}\text{s}}
\newcommand{\upcast}{\mathsf{Up}{\text -}\mathsf{cast}}
\newcommand{\downcast}{\mathsf{Down}{\text -}\mathsf{cast}}

\newcommand{\poly}{\operatorname{poly}}
\newcommand{\ID}{\operatorname{ID}}
\newcommand{\LL}{\mathcal{L}}
\newcommand{\DD}{\mathcal{D}}
\newcommand{\CC}{\mathcal{C}}
\newcommand{\SSS}{\mathcal{S}}
\newcommand{\RR}{\mathcal{R}}
\newcommand{\UU}{\mathcal{U}}

\newcommand{\idle}{\mathsf{idle}}
\newcommand{\listen}{\mathsf{listen}}
\newcommand{\transmit}{\mathsf{transmit}}
\newcommand{\Ones}{\mathsf{Ones}}
\newcommand{\Zeros}{\mathsf{Zeros}}

\newcommand{\CONGEST}{\mathsf{CONGEST}}
\newcommand{\RN}{\mathsf{RN}}
\newcommand{\Energy}{\mathsf{En}}
\newcommand{\Time}{\mathsf{Time}}

\begin{document}

\title{The Energy Complexity of BFS in Radio Networks}

 \author{Yi-Jun Chang\\
 {\small ETH Z\"{u}rich}
 \and
 Varsha Dani\\
 {\small Univ.~of New Mexico}
 \and
 Thomas P. Hayes\thanks{Supported by NSF CAREER award CCF-1150281.}\\
 {\small Univ.~of New Mexico}
 \and
 Seth Pettie\thanks{Supported by NSF grants CCF-1514383, CCF-1637546, and CCF-1815316.}\\
 {\small Univ.~of Michigan}}

\date{}
\maketitle
\thispagestyle{empty}
\setcounter{page}{0}

\begin{abstract}
We consider a model of \emph{energy complexity} in Radio Networks in which transmitting or listening on the channel costs one unit of energy and computation is free.  This simplified model captures key aspects of battery-powered sensors: that battery-life is most influenced by transceiver usage, and that at low transmission powers, the actual cost of transmitting and listening are very similar.

The energy complexity of tasks in single-hop (clique) networks are well understood~\cite{ChangKPWZ17,NakanoO00,BenderKPY18,JurdzinskiKZ02c}.  Recent work of Chang et al.~\cite{ChangDHHLP18} considered energy complexity in \emph{multi-hop} networks and showed that $\mathsf{Broadcast}$ admits an \emph{energy-efficient} protocol, 
by which we mean each of the $n$ nodes in the network spends $O(\polylog(n))$ energy.
This work left open the strange possibility that \emph{all} natural problems in multi-hop networks might admit such an energy-efficient solution.

In this paper we prove that the landscape of energy complexity is rich enough to support a multitude of problem complexities.  Whereas $\mathsf{Broadcast}$ can be solved by an energy-efficient protocol, exact computation of $\mathsf{Diameter}$ cannot, requiring $\Omega(n)$ energy.
Our main result is that 
$\mathsf{Breadth First Search}$ has 
sub-polynomial energy complexity at most 
$2^{O(\sqrt{\log n\log\log n})}=n^{o(1)}$; 
whether it admits an efficient 
$O(\polylog(n))$-energy protocol is an open problem.

Our main algorithm involves recursively solving a generalized BFS problem on a ``cluster graph'' introduced by Miller, Peng, and Xu~\cite{miller2013parallel}.  In this application, we make crucial use of a close
relationship between distances in this cluster graph,
and distances in the original network.  This relationship is new and may be of independent interest.

We also consider the problem of approximating the network $\mathsf{Diameter}$.  From our main result, it is immediate that $\mathsf{Diameter}$ can be 2-approximated using $n^{o(1)}$ energy per node.  We observe that, for all $\epsilon > 0$, approximating $\mathsf{Diameter}$ to within a $(2-\epsilon)$ factor 
requires $\Omega(n)$ energy per node.  However, this lower bound is only due to graphs of very small diameter; for large-diameter graphs, we prove that the diameter can be nearly $3/2$-approximated using $O(n^{1/2+o(1)})$ energy per node. 
\end{abstract}

\newpage

\section{Introduction}

Consider a network of $n$ tiny sensors scattered throughout a National Park.
We'd like the sensors to organize themselves, so that in the
event of a forest fire, say, information about it can be \emph{efficiently}
broadcast to the entire network.

In this extremely low power setting, sensors would need to
spend most of their time with their transceiver units shut off to
conserve power.  In a steady state, we might expect that we have a
good \emph{labelling} of the nodes, and each node with label $i$ wakes up
at times of the form $jP + i$, where $j$ runs through every positive integer, and
$P$, the polling period, is also a positive integer.  
Each node wakes up just long enough to receive a message
and forward it on any neighbors with label $i+1$.  In this
way, at the expense of adding $P$ to the latency, the nodes are able to reduce
their power consumption by a factor of $P$, compared to the always-on
scenario.

Once $P$ has been optimized, which should be a function of the
available power, the next issue is how to find a good
labelling efficiently.  In this paper we focus mainly on the problem of computing
\emph{BFS labelings}: a given source $s$ has label zero, 
and all other devices label themselves by the distance (in hops) to $s$.
Such a labeling gives a 2-approximation to the diameter, and via up-casts
and down-casts, allows for time- and energy-efficient dissemination of a message from any origin.
Thus, the problem of finding a BFS labelling is a very natural question in
this context.

\subsection{The Model}

We work within the classic \emph{Radio Network} model~\cite{chlamtac1985broadcasting},
but in contrast to most prior work in this model, we treat
\emph{energy} (defined below) as the primary measure of complexity
and \emph{time} to be important, but secondary.

There are $|V|$ devices 
associated with the nodes of an \underline{\emph{unknown}}
undirected graph $G=(V,E)$.
\emph{Time} is partitioned into discrete steps.  
All devices agree on time zero,\footnote{Synchronizing devices in an energy-efficient manner is an interesting open problem.  In some situations it makes sense to assume the devices begin in a synchronized state, e.g., if the sensors are simultaneously turned on and dropped from an airplane on the aforementioned National Park.} 
and agree on some upper bound $n\ge |V|$.
In each timestep, each device performs some computation and chooses to either 
$\idle$, $\listen$ to the channel, or $\transmit$ a message.
If a device $v$ chooses to $\listen$, 
and \emph{exactly} one device $u\in N(v)$ $\transmit$s a message $m_u$,
then $v$ receives $m_u$.  In all other cases, $v$ receives no feedback from the
environment.\footnote{Here $N(v)=\{u \ | \ \{u,v\}\in E(G)\}$ is the neighborhood of $v$. 
A more powerful model allows for \emph{collision detection}, 
i.e., differentiation between zero and two or more transmitters in $N(v)$.
Since collision detection only 
gives a $\polylog (n)$ advantage
in any complexity measure ($\sr$ in Section~\ref{sect:cluster} allows each vertex to differentiate  between zero and two or more transmitters in $\polylog (n)$ rounds w.h.p.) and we 
are insensitive to such factors, 
we assume the weakest model, without collision detection.}
Devices can locally generate unbiased random bits; there is no shared randomness.
Let $\RN[b]$ denote this Radio Network model, 
where $b$ is the maximum number of bits per message.
All of our algorithms work in $\RN[O(\log n)]$ and 
all our lower bounds apply even to $\RN[\infty]$.

\paragraph{Cost Measures.} An algorithm runs in \emph{time $t$} if all devices
halt and return their output by timestep $t$.
Typically the algorithm is randomized, with some probability of failure, but
$t$ is a function of $n$ or other given parameters, not a random variable. 
The \emph{energy cost} of $v\in V$ is the number of timesteps for which
$v$ is $\listen$ing or $\transmit$ting.  
(This is motivated by the fact that the \emph{sleep mode} of tiny devices is so efficient that it is reasonable to approximate its energy-cost by \emph{zero}, and that transceiver usage is often the most expensive part of a computation.  Moreover, at low transmission powers, transmitting and listening are comparable; see, e.g.,~\cite[Fig.~2]{PolastreSC05} and \cite[Table~1]{BarnesCMA10}.)
The energy cost of the \emph{algorithm}
is the maximum energy cost of any device.

\paragraph{Energy Complexity.} 
Most prior work on energy complexity
has focused on \emph{single-hop} (clique) networks, 
typically under the assumption that $|V|=n$ is \emph{unknown}, 
and that some type of collision-detection is available.\footnote{Sender-side CD enables devices to detect if another device is transmitting; receiver-side CD lets receivers detect if at least two devices are transmitting.}
Because of the high degree of symmetry, there are only
so many interesting problems in single-hop networks.
Nakano and Olariu~\cite{NakanoO00} proved that the $\mathsf{Initialization}$
problem (assign devices distinct IDs in $\{1,\ldots,|V|=n\}$)
can be solved with $O(\log\log n)$ energy.
Bender et al.~\cite{BenderKPY18} showed that with collision-detection, 
all $n$ devices holding messages
can transmit all of them using $O(\log(\log^* n))$ energy.
Chang et al.~\cite{ChangKPWZ17} proved that 
$\Theta(\log(\log^* n))$ is optimal,
and more generally, settled the complexity of $\mathsf{LeaderElection}$ and $\mathsf{ApproximateCounting}$ (estimating ``$n$'') in all the collision-detection models, with and without randomization.  It was proved that collision-detection gives two exponential advantages in energy complexity. With randomization, $\mathsf{LeaderElection}/\mathsf{ApproximateCounting}$ takes 
$\Theta(\log^* n)$ energy (without CD) or $\Theta(\log(\log^* n))$ energy (with CD), and deterministically, they take 
$\Theta(\log N)$ energy (without CD~\cite{JurdzinskiKZ02c})
and $\Theta(\log\log N)$ energy (with CD), 
where devices initially have IDs in $[N]$.  See also~\cite{JurdzinskiKZ02b,JurdzinskiKZ02,JurdzinskiKZ02c,JurdziskiKZ03,JurdzinskiS02}.
Three-way tradeoffs between time, energy, and error probability were 
studied by 
Chang et al.~\cite{ChangKPWZ17} and Kardas et al.~\cite{KardasKP13}.

Very recently Chang et al.~\cite{ChangDHHLP18} extended the 
single-hop notion of energy complexity
to \emph{multi-hop} networks ($G$ is \emph{not} a clique),
and proved nearly sharp upper and lower bounds on $\mathsf{Broadcast}$,
both in $\RN[O(\log n)]$ and the same model when listeners have 
collision detection.  Without CD the energy
complexity of $\mathsf{Broadcast}$ is between $\Omega(\log^2 n)$ and $O(\log^3 n)$; with CD it is between 
$\Omega(\log n)$ and $O\left(\frac{\log n\log\log n}{\log\log\log n}\right)$. 

\paragraph{Other Energy Models.}
Other notions of energy complexity have been studied in radio networks.
For example, when distances between devices are very large, transmitting is 
significantly more expensive than listening, and it makes sense to design algorithms
that minimize the worst-case number of transmissions per device.
Gasnieniec et al.~\cite{GasieniecKKPS07}, 
Klonowski and Pajak~\cite{KlonowskiP18}, 
and Berenbrink et al.~\cite{BerenbrinkCH09} studied broadcast 
and gossiping problems under this cost model.
Klonowski and Sulkowska~\cite{KlonowskiS16} defined a distributed
model in which devices are scattered randomly at points in $[n^{1/d}]^d$
and can choose their transmission power dynamically.
Several works have looked at energy complexity against
an adversarial \emph{jammer}, where the energy cost is sometimes
a function of the adversary's energy budget.  See, e.g., \cite{KutylowskiR03,KabarowskiKR06,GilbertKPPSY14,KingPSY18}.

\paragraph{Time Complexity.} Most prior work in the $\RN$ model
has studied the time complexity of basic primitives such as $\mathsf{LeaderElection}$, $\mathsf{Broadcast}$, $\mathsf{BFS}$, etc.  
We review a few results most relevant to our work. 
Bar-Yehuda et al.'s~\cite{bar1991efficient} \emph{decay} algorithm
 solves $\mathsf{BFS}$ in $O(D\log^2 n)$ time and $\mathsf{Broadcast}$
 in $O(D\log n + \log^2 n)$ time.  Here $D$ is the diameter of the network.
 Since $\Omega(D)$ is an obvious lower bound,
 the question is which $\log$-factors are necessary.  Alon et al.~\cite{alon1991lower} proved that the additive $\log^2 n$ term is necessary
 in a strong sense: even with full knowledge of the graph topology,
 $\mathsf{Broadcast}$ needs $\Omega(\log^2 n)$ time even when $D=O(1)$.
Kushilevitz and Mansour~\cite{KushilevitzM98} proved that if 
devices are forbidden from transmitting before hearing the message,
then $\Omega(D\log(n/D))$ time in necessary.
Czumaj and Davies~\cite{CzumajD17} (improving~\cite{haeupler2016faster})
gave a $\mathsf{Broadcast}$ algorithm running in $O(D\log_D n + \polylog (n))$ time, which is optimal when $D > n^\epsilon$.
These $\mathsf{Broadcast}$ algorithms \emph{do not} solve $\mathsf{BFS}$.
Improving the classic $O(D\log^2 n)$ decay algorithm for $\mathsf{BFS}$,
Ghaffari and Haeupler~\cite{GhaffariH16} solve $\mathsf{BFS}$
in $O(D\log (n)\log\log (n) + \polylog (n))$ time.

\paragraph{New Results.}
It is useful to coarsely classify energy-efficiency bounds
as either \emph{feasible} or \emph{infeasible}.  
We consider $\polylog (n)$ energy to be feasible 
and polynomial energy $n^{\Omega(1)}$ to be infeasible.\footnote{These definitions seem to be robust to certain modeling assumptions, e.g., 
whether collision detection is available.}
It is not immediately obvious that there are \emph{any} natural, 
infeasible problems, especially if we are considering the full
power of $\RN[\infty]$, where message congestion is not an issue.  
In this paper we demonstrate
that the energy landscape is rich, and that even
coarsely classifying the energy complexity of simple problems 
is technically challenging and demands the development of 
new algorithm design techniques. 
Our results are as follows
\begin{itemize}
    \item We develop a recursive $\mathsf{BreadthFirstSearch}$ algorithm in $\RN[O(\log n)]$ with ``intermediate'' energy-complexity $2^{O(\sqrt{\log n\log\log n})} = n^{o(1)}$.
    The algorithm involves simulating itself on a clustered version of the input graph.  Due to the nature of the $\RN$ model, this simulation is not free, but incurs a polylogarithmic increase in energy, which restricts the profitable depth of 
    recursion to be at most $\sqrt{\log n/\log\log n}$.
    
    \item We give examples of some ``hard'' problems in energy-complexity, even when the model is $\RN[\infty]$. 
    The problem of deciding whether $\mathsf{diam}(G)$ is 1 or at least 2 takes $\Omega(n)$ energy; in this case the hard graph $G$ is dense.  We adapt the construction of~\cite{AbboudCK16} (designed for the $\mathsf{CONGEST}$ model) to show that even on \emph{sparse} graphs, with arboricity $O(\log n)$, deciding whether $\mathsf{diam}(G)$ is 2 or at least 3 takes $\tilde{\Omega}(n)$ energy.
    %\footnote{This result first appeared in a technical report by Y.-J. Chang~\cite{Chang18} 
    %}
    %by the first author.}
    
    \item To complement the hardness results, we show that $\mathsf{Diameter}$
    can be nearly $3/2$-approximated\footnote{I.e., it returns a value in the range $\left[\floor{\frac{2}{3}\mathsf{diam}(G)}, \mathsf{diam}(G)\right]$.} in $\RN[O(\log n)]$ 
    with $O(n^{1/2+o(1)})$ energy, by adapting~\cite{holzer2014brief,RodittyW13} and
    using our new $\mathsf{BreadthFirstSearch}$ routine.
\end{itemize}

The existence of a subpolynomial-energy $\mathsf{BreadthFirstSearch}$ algorithm
is somewhat surprising for information-theoretic reasons.
Observe that the number of edges in $E(G)$ that are collectively discovered
by all devices is at most the number of messages successfully received,
which itself is at most the aggregate energy cost.  Thus, if the per-device energy cost is $n^{o(1)}$, we can never hope to know about more than $n^{1+o(1)}$ edges in $E(G)$ --- a negligible fraction of the input on dense graphs!  
On the other hand, it is possible to efficiently verify the \emph{non-existence} of
many non-edges.  Given a candidate $\mathsf{BFS}$-labeling, for example, it is straightforward to \emph{verify} its correctness with $\polylog(n)$ energy.

\paragraph{Organization.}
In Section~\ref{sect:cluster} we review the 
Miller-Peng-Xu~\cite{miller2013parallel} clustering
algorithm and prove that it preserves distances better than previously
known.  
In Section~\ref{sect:cluster-sim} we 
define some communications primitives and 
prove that they can be executed
on the cluster graph (as if it were an $\RN[O(\log n)]$ network)
at the cost of a polylogarithmic factor increase in energy usage.  
In Section~\ref{sect:BFS} we 
design and analyze
a recursive BFS algorithm,
which uses $2^{O(\sqrt{\log n\log\log n})}$ energy.
In Section~\ref{appendix:diameter} we consider the energy cost
of approximately computing the network's $\mathsf{Diameter}$.

\section{Cluster Partitioning}\label{sect:cluster}

Miller, Peng, and Xu~\cite{miller2013parallel} 
introduced a remarkably simple algorithm for 
partitioning a given graph into vertex-disjoint clusters 
with certain desirable properties.  In this section
we prove that the MPX clustering approximately preserves relative \emph{distances} from the original graph significantly
better than previously known.

Given a graph $G = (V,E)$, and a parameter $\beta$, 
each vertex $v \in V$ independently samples a random variable
$\delta_v \sim \operatorname{Exponential}(\beta)$ from the 
exponential distribution with mean $1/\beta$.
Assign each $v$ to the ``cluster'' centered at 
$u \in V$
that minimizes  $\dist_G(v,u) - \delta_u$.
Equivalently, we may think of a cluster forming at each vertex $u$
at time $-\delta_u$, and spreading through the graph at a uniform rate of
one edge per time unit.  
Each vertex $v$ is absorbed into the first 
cluster to reach it, if this happens prior to time $-\delta_v$,
when it would start growing its own cluster.
Refer to Figure~\ref{fig:cluster-graph}.
Throughout the paper, 
\emph{we only choose $\beta$ such that $1/\beta$ is an integer}.

%\input{figure-cluster-graph.tex}

%%%%%%%%%%%%%%%%%%%
%%%%%%%%%%%%%%%%%%%
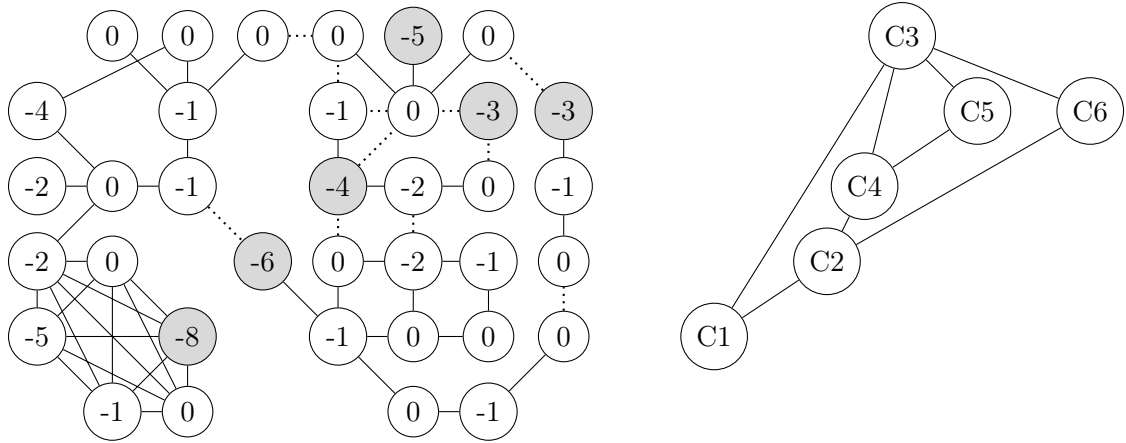
\begin{figure*} 
\begin{center}
\begin{tikzpicture}[scale=1]
%\draw (0,0) --(1,2);
%\draw[help lines] (0,0) grid (2,3);
% Outer box:
%\draw (0,0) rectangle (16,16);
\tikzstyle{every node} = [draw, circle]
\node (a) at (2,1){-1};
\node (b) at (3,1){0};
\node (c) at (1,2){-5};
\node [fill=gray!30](d) at (3,2){-8};
\node (e) at (1,3){-2};
\node (f) at (2,3){0};
\draw (a)--(b)--(c)--(d)--(e)--(f)--(a)--(c)--(e)--(a)--(d);;
\draw (b)--(d)--(f)--(b)--(e);
\draw (c)--(f);
\node (g) at (1,4){-2};
\node (h) at (1,5){-4};
\node (i) at (2,4){0};
\node (j) at (2,6){0};
\node (k) at (3,4){-1};
\node (l) at (3,5){-1};
\node (m) at (3,6){0};
\draw (e)--(i)--(g); \draw (i)--(h);
\draw (i)--(k)--(l)--(m); \draw (l)--(j);
\draw (m)--(h);
\node [fill=gray!30] (n) at (4,3){-6};
\node (o) at (4,6){0};
\node (p) at (5,2){-1};
\node (q) at (5,3){0};
\node [fill=gray!30] (r) at (5,4){-4};
\node (s) at (5,5){-1};
\node (t) at (5,6){0};
\draw [dotted, thick] (k)--(n);
\draw (n)--(p)--(q);
\draw (r)--(s);
\draw [dotted, thick] (s)--(t);
\draw [dotted, thick] (o)--(t);
\draw (o)--(l);
\node (u) at (6,1){0};
\node (v) at (6,2){0};
\node (w) at (6,3){-2};
\node (x) at (6,4){-2};
\node (y) at (6,5){0};
\node [fill=gray!30] (z) at (6,6){-5};
\node (aa) at (7,1){-1};
\node (ab) at (7,2){0};
\node (ac) at (7,3){-1};
\node (ad) at (7,4){0};
\node [fill=gray!30] (ae) at (7,5){-3};
\node (af) at (7,6){0};
\draw (p)--(v)--(ab)--(ac)--(w)--(q);
\draw (r)--(x)--(ad);
\draw [dotted, thick] (q)--(r);
\draw [dotted, thick] (ad)--(ae)--(y)--(s);
\draw (v)--(w); \draw (y)--(z);
\draw [dotted, thick] (w)--(x);
\node (ag) at (8,2){0};
\node (ah) at (8,3){0};
\node (ai) at (8,4){-1};
\node [fill=gray!30] (aj) at (8,5){-3};
\draw (p)--(u)--(aa)--(ag);
\draw (ah)--(ai)--(aj);
\draw (af)--(y);
\draw [dotted, thick] (aj)--(af);
\draw [dotted, thick] (ag)--(ah);
\draw (y)--(t);
\draw [dotted, thick] (r)--(y);

\node (C1) at (10,2){C1};
\node (C2) at (11.5,3){C2};
\node (C3) at (12.5,6){C3};
\node (C4) at (12,4){C4};
\node (C5) at (13.5,5){C5};
\node (C6) at (15,5){C6};
\draw (C2)--(C1)--(C3);
\draw (C3)--(C5)--(C4)--(C3);
\draw (C4)--(C2)--(C6)--(C3);
% \node [fill=gray!30] (a) at (5,5) {-9};
% \node (b) at (4,7) {-2};
% \node (c) at (2,3) {-4};
% \node (d) at (9,8) {-1};
% \draw (a)--(b)--(c)--(9,8);
% % Tick marks:
%\draw (1,0)--(1,2);
\end{tikzpicture}
\end{center}
\caption{Constructing a cluster graph.  At left, the original graph;
  with the (rounded) start time, $-\delta_v$, marked on each vertex.  
The cluster centers have been darkened, and the dotted lines
indicate edges that cross a cluster boundary.
At right, the corresponding cluster graph.  Note that the distances in
the cluster graph are broadly proportional to the original distances,
but can vary significantly.}
\label{fig:cluster-graph}
\end{figure*}

%%%%%%%%%%%%%%%%%%%
%%%%%%%%%%%%%%%%%%%

Miller et al.~\cite{miller2013parallel} were 
primarily interested in this construction because 
the algorithm parallelizes well, the clusters have diameter
$O(\log(n)/\beta)$ w.h.p., and a $O(\beta)$-fraction of the
edges are ``cut,'' having their endpoints in distinct clusters.
Haeupler and Wajc~\cite{haeupler2016faster} observed 
that this 
algorithm can be efficiently implemented in the Radio Network model~\cite{chlamtac1985broadcasting,ChlamtacK87}, 
with only minor modifications.  

\subsection{The Cluster Graph as a Distance Proxy}

Define $\cluster(u)$ to be the cluster containing $u$.
The cluster graph, $\clusterGraph{G}{\beta} = G^* = (V^*,E^*)$
is defined by
\begin{align*}
             V^*    &= \{\cluster(u) \ | \ u\in V(G)\}\\
\mbox{ and } E^*    &= \{(\cluster(u),\cluster(v)) \ | \ (u,v)\in E(G), \cluster(u)\neq\cluster(v)\}.
\end{align*}

To prove that distances in $G^*$ are a good proxy
for distances in $G$, we make use of the
following lemma, which is a slight variant of 
lemmas by Miller, Peng, Vladu, and Xu~\cite[Lemma 2.2]{MillerPVX15}
and Haeupler and Wajc~\cite[Corollary 3.8]{haeupler2016faster}.  
We include a proof for completeness.

Define
$\ball{G}{v}{\ell}
=\{ u \in V \ | \ \dist_G(u,v) \leq \ell\}$
to be the ball of radius $\ell$ around $v$.

\begin{lemma} \label{lem:cluster-exponential-hits-in-ball}
Let $G^* = \clusterGraph{G}{\beta}$ be the cluster graph 
for $G$.  For every positive integer $j$ and $\ell>0$, 
the probability that the number of $G^*$-clusters 
intersecting $\ball{G}{v}{\ell}$
is more than $j$ is at most
\[
(1 - \exp(-2 \ell \beta))^j.
\]
\end{lemma}

\begin{proof}
Condition on the time $t$ that the $(j + 1)$st signal would reach vertex
$v$, as well as on the identities $v_1, \dots, v_j$ 
of the vertices whose signals reach $v$ before time $t$.  
Due to the memoryless property of the exponential distribution, each of these arrival
times are independently distributed as 
$\min\{t, \dist(v_i, v)\} - X \leq t - X$, where $X \sim \operatorname{Exponential}(\beta)$.

Now, if $\max_{1 \le i \le j} X_i > 2\ell$, then $\ball{G}{v}{\ell}$
cannot intersect any clusters except those 
centered at $v_1, \dots,
v_j$, because they do not 
reach $\ball{G}{v}{\ell}$ until times 
$\ge t - \ell$, whereas the first signal reached $v$ before time $t - 2\ell$, and has therefore already 
flooded all of $\ball{G}{v}{\ell}$
before time $t - \ell$.  Thus,
\[
\Prob{\ball{G}{v}{\ell} \mbox{ intersects more than $j$ clusters}} \le \Prob{\forall i\in[1,j], X_i \le 2 \ell} = (1 - \exp(-2 \ell \beta))^{j}.   \qedhere
\]
%\[
%\Prob{\ball{G}{v}{\ell} \mbox{ intersects more than $j$ clusters}} \le \Prob{\forall i\in[1,j], X_i \le 2 \ell}
%= (1 - \exp(-2 \ell \beta))^{j}.\qedhere
%\]
\end{proof}

A natural way to show that $G^*$ approximately
preserves distances in $G$ is to consider the fraction 
of edges in a shortest path that are ``cut'' by 
the partition, which corresponds to applying Lemma~\ref{lem:cluster-exponential-hits-in-ball} with $\ell=1/2$ and $j=1$.\footnote{One imagines a vertex $v_e$
in the middle of an edge $e$; $e$ is cut iff $\ball{G}{v_e}{1/2}$ intersects two clusters,
which must cover distinct endpoints of $e$.}
This was the approach taken in~\cite{ChangDHHLP18},
but it only guarantees that the fraction of edges cut 
concentrates around its expectation ($O(\beta)$)
for paths of length $\tilde{\Omega}(\poly(\beta^{-1}))$.
In Lemmas~\ref{lem:dist-ratio-union-bound} and \ref{lem:dist-ratio-union-bound2}
we use Lemma~\ref{lem:cluster-exponential-hits-in-ball}
in a different way to bound 
the ratio of distances in $G$ to those in $G^*$,
which works even for relatively short distances.
Lemma~\ref{lem:dist-ratio-union-bound} applies to all distances
(and suffices for our BFS application in Section~\ref{sect:BFS})
whereas Lemma~\ref{lem:dist-ratio-union-bound2} applies 
to distances $\Omega(\beta^{-1}\log^2 n)$.

\begin{lemma} \label{lem:dist-ratio-union-bound}
Let $G^* = \clusterGraph{G}{\beta}$ be a clustering of $G$.
There exists a constant $C$ such that for every pair $u,v \in V(G)$,
\begin{dmath*}
\Prob{\dist_{G^*}(\cluster(u),\cluster(v))
 \in \left[\floor{\frac{\dist_{G}(u,v)\cdot \beta}{8\log(n)}}, 
\ceil{\dist_G(u,v)\cdot\beta} \cdot C\log(n)\right]} 
\ge 1 - \frac{1}{n^3}.  
\end{dmath*}
%\[
%\Prob{\dist_{G^*}(\cluster(u),\cluster(v))
%\in \left[\floor{\frac{\dist_{G}(u,v)\cdot \beta}{8\log(n)}}, \ 
%\ceil{\dist_G(u,v)\cdot\beta} \cdot C\log(n)\right]} 
%\ge 1 - \frac{1}{n^3}.
%\]
More generally, let $P=(u,\dots,v)$ be any length-$d$ path connecting $u$ and $v$. With probability $1 - \frac{1}{n^3}$, there exists a path $P^*$ 
in $G^*$ connecting $\cluster(u)$ and $\cluster(v)$ with length at most
$d\cdot C\beta\log(n)$, where each cluster in $P^*$ intersects $P$.
\end{lemma}

\begin{proof}
First observe that the probability of any $\delta_v$-value being 
outside $[0,4\log(n)/\beta)$ is $\ll n^{-4}$ and hence all clusters 
have radius less than $4\log(n)/\beta$ with probability $\ll n^{-3}$.
This gives the lower bound on $\dist_{G^*}(u,v)$. 

For the upper bound, define $\ell$ to be the integer $1/\beta$.  
Fix any length-$d$ path $P$ from $u$ to $v$ (e.g., a shortest
path, with $d=\dist_G(u,v)$), 
and cover its vertices
with $\left\lceil \frac{d}{2 \ell + 1} \right\rceil$ paths of length
$2 \ell$.  
Applying Lemma~\ref{lem:cluster-exponential-hits-in-ball}
to the center vertex $u'$ of one of these subpaths, we conclude
that the number of clusters that intersect $\ball{G}{u'}{\ell}$,
(which includes the entire subpath) is more than $j$ with probability
\begin{equation}\label{eqn:clusters-geom-distr}
\left(1 - \exp(-2 \beta \ell) \right)^j 
=
(1 - \exp(-2))^j,
\end{equation}
Choosing $j$ to be the appropriate
multiple of $\log(n)$, we can make this probability $\ll n^{-4}$.  
Taking a union bound over the
$\approx \beta d/2 < n$ subpaths, 
the probability that any subpath intersects more than $C \log(n)$
clusters is $\ll n^{-3}$.  This concludes the proof.
\end{proof}

Lemma~\ref{lem:dist-ratio-union-bound} suffices
to achieve our main result, BFS labeling in $2^{O(\sqrt{\log n\log\log n})}$ energy, but 
the exponent can be improved
by a constant factor by using Lemma~\ref{lem:dist-ratio-union-bound2} whenever applicable.  We include the proof of Lemma~\ref{lem:dist-ratio-union-bound2} 
since it may be of independent interest.

\begin{lemma}\label{lem:dist-ratio-union-bound2}
Let $G^* = \clusterGraph{G}{\beta}$ be a clustering of $G$.
There exists a constant $C$ such that for every pair $u,v \in V(G)$
\begin{align*}
&\Prob{\dist_{G^*}(\cluster(u),\cluster(v))
\in \left[\frac{\dist_{G}(u,v)\cdot \beta}{8\log(n)}, \ \dist_G(u,v)\cdot C\beta\right]}
\ge 1 - \frac{1}{n^3}.
\end{align*}
%\[
%\Prob{\dist_{G^*}(\cluster(u),\cluster(v))
%\in \left[\frac{\dist_{G}(u,v)\cdot \beta}{8\log(n)}, \ \dist_G(u,v)\cdot C\beta\right]} 
%\ge 1 - \frac{1}{n^3}.
%\]
\end{lemma}

\begin{proof}
We condition on the event that all cluster radii are at most $4\log(n)/\beta$, which fails to hold with probability $\ll n^{-3}$.
As before, the lower bound on $\dist_{G^*}(u,v)$ follows from this event.
Furthermore, this implies that sufficiently distant segments on the shortest $u$-$v$ path are essentially independent.

As before, cover the vertices of the shortest $u$-$v$ path with 
length-$2\ell$ subpaths, $\ell=1/\beta$, 
and color the subpaths with $4\log(n)+1$
colors such that any two subpaths of the same color are at distance at least
$8\log(n)/\beta$.  Each color-class contains $\Omega(\log n)$ subpaths.
By Lemma~\ref{lem:cluster-exponential-hits-in-ball} and (\ref{eqn:clusters-geom-distr}), 
the number
of clusters intersecting subpaths of a particular color class is stochastically
dominated by the sum of $\Omega(\log n)$ geometrically distributed 
random variables with constant expectation 
$\frac{1}{1-(1-\exp(-2))} = \exp(2)$.  
By a Chernoff bound, the probability that this sum deviates from its 
expectation by more than a constant factor is 
$1/\poly(n)$.
% \footnote{Chernoff bounds cannot be applied \emph{directly} to
% sums of geometric random variables.  However, it is straightforward to 
% bound the upper tail by considering a related sequence of Bernoulli trials
% with success probability $\exp(-2)$.}
Hence, for sufficiently large $C$ 
(controlling the number of summands and the tolerable deviation)
the probability that any color-class hits too many distinct clusters 
is $\ll n^{-3}$.
\end{proof}

\begin{remark}
Lemma~\ref{lem:dist-ratio-union-bound2} cannot be improved by 
more than constant factors.   
It is easy to construct families of graphs for which both the upper and
lower bounds are tight, with high probability, depending on which 
vertex pairs are chosen.
\end{remark}

\subsection{Distributed Implementation}

The definition of $\clusterGraph{G}{\beta}$ immediately 
lends itself to a distributed implementation 
in radio networks, as was noted in~\cite{haeupler2016faster}.
For completeness we show how it can be reduced
to calls to $\sr$.

\begin{description}
\item [$\sr$:] 
We are given two disjoint vertex sets $\SSS$ and $\RR$,
where each vertex $u \in  \SSS$ holds a message $m_u$.
An $\sr$ algorithm guarantees that for
every $v \in \RR$ with $N(v) \cap \SSS \neq \emptyset$,
with probability $1-f$, $v$ receives some message
$m_u$ from {\em at least one} vertex $u \in N(v) \cap \SSS$. 
We only apply this routine with $f = 1/\poly(n)$.
\end{description}

\begin{lemma}\label{lemma:sr-decay}
$\sr$ can be implemented in $O(\log\Delta\log f^{-1})$ time and energy,
where $\Delta\leq n-1$ is an upper bound on the maximum degree.
Senders use $O(\log f^{-1})$ energy;
receivers that hear a message use $O(\log \Delta)$ energy in expectation;
receivers that hear no message use $O(\log\Delta\log f^{-1})$ energy.
\end{lemma}

\begin{proof}
This lemma follows from a small modification to the \emph{Decay} algorithm~\cite{bar1992time}, which is known to be optimal 
in terms of \emph{time}; see Newport~\cite{Newport14}.  
For the sake of completeness, 
we provide a proof here.  Each sender $u\in \SSS$ 
repeats the following $O(\log f^{-1})$ 
times.  Randomly pick an $X_u \in [1,\log\Delta]$ 
such that $\Prob{X_u =t} \ge 2^{-t}$ and transmit $m_u$ at time step $X_u$.  
The energy of any sender is clearly $O(\log f^{-1})$ with probability 1.
For a receiver $v\in \RR$, if the number of senders in 
$N(v)$ is in the range $[2^{t-1},2^t]$, $v$ will receive \emph{some}
message with constant probability in the $t$th timestep of every iteration.
Receivers with no adjacent sender will never detect this, and
spend $\Theta(\log\Delta\log f^{-1})$ energy.
\end{proof}

We show that $\clusterGraph{G}{\beta}$ can be computed, w.h.p., 
using $4\log(n)/\beta$ $\srs$ in the communication network $G=(V,E)$.
Every vertex $u$ will learn its cluster-identifier $\ID(\cluster(u)))$ and get a label $\LL(v)$
such that $\LL(v)=0$ iff $v$ is a cluster center
and $\LL(v)=i$ iff there is a $u\in N(v)$ with $\LL(u)=i-1$ such that $\cluster(u)=\cluster(v)$.
If $\LL(v) = i$, we say that $v$ is at \emph{layer $i$}. 

The graph $\clusterGraph{G}{\beta}$ is constructed as follows.
Every vertex $v$ picks a value $\delta_v\sim \text{Exponential}(\beta)$
and sets its start time to be
$\text{start}_v\gets \lceil \frac{4 \log(n)}{\beta} - \delta_v\rceil$.  With probability at least $1-1/n^3$, all start times are positive.
For $i = 1$ to $4\log(n)/\beta$, do the following. 
At the beginning of the $i$th iteration, if $v$ is not yet in any cluster and
$\text{start}_v=i$, then $v$ becomes a cluster center
and sets $\LL(v)=0$.
During the $i$th iteration, we execute $\sr$
with $\SSS$ being the set of all clustered vertices and 
$\RR$ the set of all as-yet unclustered vertices.
The message of $u\in \SSS$ contains $\ID(\cluster(u))$ and $\LL(u)$.
Any vertex $v \in \RR$ receiving a message from $u \in \SSS$ joins 
$u$'s cluster and sets $\LL(v)=\LL(u)+1$.
Lemma~\ref{lem:clustr-diam-ub} follows immediately from the 
above construction.

\begin{lemma}\label{lem:clustr-diam-ub}
The cluster graph $\clusterGraph{G}{\beta}$ can be constructed using 
$4\log(n)/\beta$ $\srs$ with probability $1-1/n^3$.  This
takes $O(\log^3(n)/\beta)$ time and 
$O(\log^3(n)/\beta)$ energy per vertex.
\end{lemma}

\section{Communication Primitives for the Cluster Graph}\label{sect:cluster-sim}

Our BFS algorithm forms a cluster graph $G^*$
and computes BFS recursively on numerous subgraphs of $G^*$.
In order for this type of recursion to work, we need to argue
that algorithms on the (abstract) $G^*$ can be simulated, with some time and energy cost, on the underlying $G$.
We focus on algorithms that are composed \emph{exclusively} of
calls to $\sr$ (as our BFS algorithm is), but the method can be used to simulate arbitrary radio network algorithms.

We use the primitives $\downcast$ and $\upcast$ to allow
cluster centers to disseminate information to their constituents
and gather information from some constituent.

\begin{description}
\item[$\downcast$:] There is a set $\UU$ of vertices such that each $u \in \UU$ is a cluster center, and the goal is to let each $u \in \UU$ broadcast a message $m_u$ to all members of $\cluster(u)$.
\item[$\upcast$:] There is a set $\UU$ of vertices such that each $u \in \UU$ wants to deliver a message $m_u$ to the center of $\cluster(u)$. 
Any cluster center $v$ with at least one $u\in \UU\cap \cluster(v)$
must receive \emph{any} message from one such vertex.
\end{description}

\begin{lemma}\label{lemma:sr-cluster}
$\upcast$ and $\downcast$ can be implemented
with $O\parens{\frac{\log^3 n}{\beta\log(1/\beta)}}$ calls
to $\sr$ on $G$, in which each vertex participates in 
$O(\log n)$ $\srs$.  I.e., the total time and energy per vertex
are $O\parens{\frac{\log^5 n}{\beta\log(1/\beta)}}$ 
and $O(\log^3 n)$, respectively.
\end{lemma}

\begin{proof}
Consider the following two quantities:
\begin{description}
\item[] $\CC = O(\log_{(1/\beta)} n)$. By Lemma~\ref{lem:cluster-exponential-hits-in-ball}, $\CC$ is an upper bound on the number of clusters intersecting $N(v) \cup \{v\}$, with high probability. 
Intuitively, $\CC$ represents the {\em contention} at $v$.
\item[] $\DD = 4\log(n)/\beta$ is the maximum radius of any cluster, i.e., the maximum $\LL$-value is at most $\DD$.
\end{description}

If there were only \emph{one} cluster, then doing an $\upcast$
or $\downcast$ would be easily reducible to $O(\log(n)/\beta)$ $\srs$.
In order to minimize interference between neighboring clusters,
we modify, slightly, the clustering algorithm so that all constituents of a cluster have shared randomness.  
When a new cluster center $v$
is formed, it generates a subset $S_{\cluster(v)} \subset [\ell]$, 
$\ell = \Theta(\CC\log n)$, 
by including each index independently with probability $1/\CC$. 
It disseminates $S_{\cluster(v)}$ to all 
members of $\cluster(v)$ along with $\ID(\cluster(v))$.
It is straightforward to show that with probability $1-1/\poly(n)$, 
for every $v$,
\begin{equation}\label{eqn:isolation}
\text{There exists } j\in[\ell] : j\in S_{\cluster(v)} 
\text{ and for all } u\in N(v), j\not\in S_{\cluster(u)} 
\end{equation}
$\downcast$ is implemented in $\DD$ stages, each stage consisting
of $\ell$ steps.  In step $j$ of stage $i$, we execute $\sr$
with $\SSS$ consisting of every $v$ with a message to send
such that $\LL(v)=i-1$ and $j\in S_{\cluster(v)}$,
and with $\RR$ consisting of every $u$ with $\LL(u)=i$
and $j\in S_{\cluster(u)}$.  By (\ref{eqn:isolation}), 
during stage $i$, every layer-$i$ vertex in every participating
cluster receives the cluster center's message with high probability.
An $\upcast$ is performed in an analagous fashion.

Each $\upcast$/$\downcast$ performs $\ell\DD=\Theta(\CC\DD\log n)=O(\frac{\log^3 n}{\beta\log(1/\beta)})$ 
%executions  of 
$\sr$ on $G$, for a total of 
$O(\frac{\log^5 n}{\beta\log(1/\beta)})$ time. 
Each vertex $v$ participates in $O(|S_{\cluster(v)}|)$
$\srs$, which is $O(\log n)$ w.h.p., for a total
of $O(\log^3 n)$ energy.
\end{proof}

\begin{lemma}\label{lem:sr-sim}
A call to $\sr$ on the cluster graph $G^* = \clusterGraph{G}{\beta}$
can be simulated with $O\left(\frac{\log^3 n}{\beta\log(1/\beta)}\right)$ calls to $\sr$ on $G$;
each vertex in $V(G)$ participates in $O(\log n)$ $\srs$.
\end{lemma}

\begin{proof}
Let $\SSS$ and $\RR$ be the sets of sending and receiving clusters in $G^*$.
All members of $C$ know that $C$ is in $\SSS$ or $\RR$.  
The $\sr$ algorithm has three steps. 
\begin{enumerate}
    \item Begin by doing a $\downcast$ in each $C\in \SSS$.  Each member of $C$ learns the 
    message $m_C$.
    \item Perform one $\sr$ on $G$, with sender set $\bigcup_{C\in \SSS} C$ and receiver
    set $\bigcup_{C'\in \RR} C'$.  At this point, w.h.p., every $\RR$-cluster adjacent to an $\SSS$-cluster has at least one constituent that has received a message.
    \item Finally, do one $\upcast$ on every cluster $C\in \RR$ to let the cluster center 
    of $C$ learn one message from a constituent of $C$, if any.
\end{enumerate}
The algorithm clearly satisfies the requirement of $\sr$ on $\clusterGraph{G}{\beta}$.
The number of calls to $\sr$ on $G$ is 
$O\left(\CC\DD\log n\right) = O\left(\frac{\log^3 n}{\beta\log(1/\beta)}\right)$
and each vertex participates in $O(\log n)$ of them.
\end{proof}

\section{BFS with Sub-polynomial Energy}\label{sect:BFS}

\subsection{Technical Overview}

Suppose every vertex in the graph could cheaply compute 
its distance from the source up to an 
\underline{\emph{additive}} $\pm \rho$ error.
Given this knowledge, 
we could trivially solve exact BFS in $\tilde{O}(D)$ time
and $\tilde{O}(\rho)$ energy per vertex, simply
by letting vertices sleep through steps that they need
not participate in.  In particular, we would advance
the \emph{BFS wavefront} one layer at a time using calls 
to $\sr$, except that each vertex $u$ would sleep through the
first $\widetilde{\dist_G}(s,u)-\rho$ calls to $\sr$,
where $\widetilde{\dist_G}$ is the approximate distance.
It would be guaranteed to fix $\dist_G(s,u)$ (and halt)
in the next $2\rho$ calls to $\sr$.  

Lemmas~\ref{lem:dist-ratio-union-bound}
and \ref{lem:dist-ratio-union-bound2} suggest a method of
obtaining approximate distances.  If we computed the cluster graph 
$G^* = \clusterGraph{G}{\beta}$ and then computed
\emph{exact} distances on $G^*$, Lemmas~\ref{lem:dist-ratio-union-bound}
and \ref{lem:dist-ratio-union-bound2} allow us to 
approximate all distances from the source, up 
to an additive error of $\tilde{O}(\beta^{-1})$ (for small distances)
and multiplicative error of $w^2$ (for larger distances),
where $w = \Theta(\log n)$ is a sufficiently large multiple of $\log n$.
Note that, from the perspective of energy efficiency, the main 
advantage to computing distances in $G^*$ rather than $G$ is
that $G^*$ has a smaller diameter $w\beta\cdot\diam(G)$.

Our algorithm computes distances up to $D$ by advancing 
the \emph{BFS wavefront} in 
$\ceil{\beta D}$ stages, extending the radius $\beta^{-1}$ per stage.
The \emph{$i$th wavefront} $W_i$ is defined to be the vertex set
\[
W_i = \{u\in V(G) \mid \dist_G(S,u) = i\beta^{-1}\},
\]
where $S$ is the set of sources. (Recall that $\beta^{-1}$ is an integer.)  
To implement the $i$th stage correctly it suffices to activate 
a vertex set $X_i$ that includes all the affected vertices, in particular:
\[
X_i \supset \{u \in V(G) \mid \dist_G(S,u) \in [i\beta^{-1},(i+1)\beta^{-1}]\} \ \ \mbox{ (w.h.p.)}
\]
In order for each vertex $u$ to decide whether it should join $X_i$ or sleep through the $i$th stage, $u$ maintains lower and upper bounds on its distance to the $i$th wavefront, or more accurately, the distance
from its cluster $\cluster(u)$ to $W_i$ in $G$.
\begin{invariant}\label{inv:interval}
Before the $i$th stage begins, each vertex $u$ knows
$L_i(\cluster(u))$ and $U_i(\cluster(u))$ such that
\[
\dist_G(W_i,\cluster(u)) = \dist_G(S,\cluster(u)) - i\beta^{-1} 
\in [L_i(\cluster(u), U_i(\cluster(u))].
\]
\end{invariant}
Clearly, if some cluster $C$ satisfies 
Invariant~\ref{inv:interval} at stage $i-1$ 
with the interval $[L_{i-1}(C),U_{i-1}(C)]$,
it also satisfies Invariant~\ref{inv:interval} at stage $i$
with $L_i(C) = L_{i-1}(C) - \beta^{-1}$ and $U_i(C) = U_{i-1}(C) - \beta^{-1}$ since the $(i-1)$th stage advances the wavefront by exactly $\beta^{-1}$. 
In the algorithm these are called \emph{Automatic Updates};
they can be done locally, without expending any energy.
In order to keep the interval $[L_i(C),U_i(C)]$ relatively narrow
(and hence useful for keeping vertices in $C$ asleep), we occasionally
refresh it with a \emph{Special Update}.  
Let $W_i^* \subseteq V(G^*)$ be the clusters in $G^*$ that intersect
the wavefront $W_i$.  We call BFS on a subgraph $G_i^*$ of $G^*$ 
from the source-set $W_i^*$, up to a radius of $Z[i]$.  
The only clusters 
that participate in this recursive call are those that are likely
to be relevant, i.e., those $C$ for which $L_i(C) \le Z[i]\cdot \beta^{-1}$.  
(The $Z[i]$ sequence will be defined shortly.)
After this recursive call completes we update
$[L_i(C),U_i(C)]$ for all participating $C$ by
applying Lemmas~\ref{lem:dist-ratio-union-bound} and \ref{lem:dist-ratio-union-bound2} 
to the (exact) distance 
$\dist_{G_i^*}(W_i^*, C)$ obtained 
in the cluster graph.

\begin{figure*}
    \centering
    \framebox{
    %\begin{minipage}{6.6in}
    \begin{minipage}{5.8in}
$\bfs(G,S,A,D)$
    \begin{enumerate}
        \item[] \hfill {\bf [Initialize Distance Estimates]}
        \item[1.] Call $\bfs(G^*,S^*,A^*,D^*)$ where $D^* = w\beta D$.  
        For each cluster $C$ in $A^*$, 
        \begin{align*}
            L_0(C) &\leftarrow \dist_{A^*}(S^*,C)\cdot \frac{1}{\beta w},
            &U_0(C) &\leftarrow \max\left\{w\beta^{-1}, w^2\cdot L_0(C)\right\}.
        \end{align*}
        \item[2.] $A = A \backslash \{u \mid L_0(\cluster(u)) = \infty\}$. \hfill {\small \emph{(Deactivate vertices at distance greater than $D$, w.h.p.)}}
        \item[3.] For $i$ from $0$ to $\ceil{\beta D}-1$ 
        \item[] \hfill {\bf [Iteratively Advance BFS Wavefront $\beta^{-1}$ Steps]}
        \begin{enumerate}
            \item[4.] Define $X_i = \{u \in A \mid L_i(\cluster(u)) \leq \beta^{-1}\}$.
            \item[5.] Advance BFS wavefront from $W_i$ to $W_{i+1}$ using $\beta^{-1}$ calls to $\sr$. Only vertices in $X_i$ participate in this step.
            \item[6.] $A \leftarrow A \backslash \{u \mid \dist_G(S,u) < (i+1)\beta^{-1}\}$.\hfill \emph{(Deactivate settled vertices.)}
            \item[] \hfill {\bf [Estimate Distances to $(i+1)$th Wavefront $W_{i+1}$]}
            \item[7.] Define $G_{i+1}^*$ to be the subgraph of $G^*$ induced by
            \[
            \Upsilon = \{C \in A^* \mid L_{i}(C) \leq (Z[i+1]+1)\cdot \beta^{-1}\}.
            \]
            Vertices in $\Upsilon$-clusters participate in a 
            \emph{Special Update}.
            
            Call $\bfs(G^*,W_{i+1}^*,\Upsilon,Z[i+1])$.
            
            For each cluster $C$ with 
            $\dist_{G_{i+1}^*}(W_{i+1}^*,C) = x$, set
            \begin{align*}
                L_{i+1}(C) &\leftarrow \min\{Z[i+1]\cdot \beta^{-1} + 1,\, x\cdot \beta^{-1}/w\},\\
                U_{i+1}(C) &\leftarrow \min\{U_{i}(C) - \beta^{-1},\, \max\{x,1\}\cdot \beta^{-1}w\}.
            \end{align*}
            \item[8.] Active vertices that did not participate in the Special Update perform an \emph{Automatic Update}. For each $C\in A^* \backslash \Upsilon$,
                \begin{align*}
                    L_{i+1}(C) &\leftarrow L_i(C) - \beta^{-1}\\
                    U_{i+1}(C) &\leftarrow U_i(C) - \beta^{-1}
                \end{align*}
        \end{enumerate}
    \end{enumerate}
    \end{minipage}
    }
    \caption{$\bfs$.}
    \label{fig:BFS}
\end{figure*}

\paragraph{Specification.}
Our $\bfs$ procedure (see Figure~\ref{fig:BFS}) takes four parameters: $G$, 
the graph, $S \subset V(G)$, the set of sources, $A \subseteq V(G)$, 
the set of \emph{active} vertices (which is a superset of $S$), 
and $D$, the depth of the search.  When we make a call to $\bfs$, 
every vertex can locally calculate $D$ and whether it is in $S$ or 
$A$.\footnote{The purpose of the $A$ parameter is to refrain
from computing useless information. E.g., when we compute the distance
from the clusters $W_i^*$ intersecting the $i$th wavefront, we are only
interested in distances to clusters intersecting as-yet unvisited vertices
(those intersecting $A$), not settled vertices ``behind'' the wavefront.}
$G^*$ denotes the cluster graph returned by $\clusterGraph{G}{\beta}$, where
$\beta$ is a parameter fixed throughout the computation.  
We compute $G^*$ once, just before the first recursive call to $\bfs(G,\cdot,\cdot,\cdot)$;
subsequent calls to $\bfs$ on $G$ with different $(S,A,D)$ parameters can use the same $G^*$. 
It is important to remember that $G$ can be either the actual radio network (RN)
or a \emph{virtual} RN on which we can simulate RN algorithms, with a certain
overhead in terms of time and energy.
At the termination of $\bfs(G,S,A,D)$, every vertex $u \in A$ returns
$\dist_{A}(S,u)$ if it is at most $D$, and $\infty$ otherwise.
Vertices in $V(G)\backslash A$ expend no energy.

\paragraph{Correctness.}
If one believes that the algorithm (Figure~\ref{fig:BFS})
faithfully implements the high level description given so far,
its correctness is immediate.
Every time we set $[L_i(C),U_i(C)]$
the interval is correct with probability $1-1/\poly(n)$, 
either because $[L_{i-1}(C),U_{i-1}(C)]$ is correct (an Automatic Update),
or because they are set according to 
Lemmas~\ref{lem:dist-ratio-union-bound} 
and \ref{lem:dist-ratio-union-bound2}, which 
hold with probability $1-1/\poly(n)$ 
(Special Update).
If $L_i(C)$ is correct for all $C$, then $X_i$ will include 
all vertices necessary to compute the $(i+1)$th wavefront,
and the $i$th stage will succeed, up to the $1/\poly(n)$
error probability inherent in calls to $\sr$.  
The main question is whether
the procedure is \emph{efficient}.

\paragraph{Efficiency.}
We will argue that for a very specific $Z[\cdot]$ sequence, which guides
the Special Update steps, the following claims hold:
\begin{claim}\label{claim:X}
Each vertex is included in the set $X_i$ for $\tilde{O}(1)$ values of $i$.
\end{claim}
\begin{claim}\label{claim:U}
    For each vertex $u$, 
    $\cluster(u)$ is included in $G_i^*$ for $\tilde{O}(1)$ values of $i$.
\end{claim}
Our algorithms ($\mathsf{cluster}$ and $\bfs$) are based solely on calls
to $\sr$.  Define $\Energy(D)$ to be the number of calls to $\sr$
that one vertex participates in when computing BFS to distance $D$.
If Claims~\ref{claim:X} and \ref{claim:U} hold, then 
\begin{equation}\label{eqn:EnD1}
\Energy(D) \leq \tilde{O}(1)\cdot \Energy(\tilde{O}(\beta D)) + \tilde{O}(\beta^{-1})
\end{equation}
The $\tilde{O}(\beta^{-1})$ term accounts for the cost of computing $G^*$ (Lemma~\ref{lem:clustr-diam-ub}) and the 
$\tilde{O}(1)$ times a vertex is included
in $X_i$ (Claim~\ref{claim:X}), 
each of which involves $\beta^{-1}$ $\sr$s on $G$.
Every recursive call to 
$\bfs(G^*,\cdot,\cdot,D')$ has $D' = \tilde{O}(\beta D)$
and by Claim~\ref{claim:U} 
each vertex participates in $\tilde{O}(1)$ such recursive calls.
Moreover, according to Lemma~\ref{lem:sr-sim}, the \emph{energy} overhead 
for simulating one call to $\sr$ on $G^*$ is $\tilde{O}(1)$ 
calls to $\sr$ on $G$.  This justifies the first term of (\ref{eqn:EnD1}).
The time and energy of our algorithm is analyzed in Theorem~\ref{thm:main}.
As a foreshadowing of the analysis, if $D_0$ is the distance
threshold of the top-level call to $\bfs$, we will set
set $\beta = 2^{-\sqrt{\log D_0 \log\log n}}$ and apply (\ref{eqn:EnD1})
to recursion depth $\sqrt{\log D_0/\log\log n}$.

\paragraph{The $Z$-Sequence.}
The least obvious part of the $\bfs$ algorithm is the $Z$-sequence,
which guides how Special Updates are performed.
Recall that $w = \Theta(\log n)$ is a sufficiently large
multiple of $\log n$; if we are computing BFS to distance $D$ in $G$,
then we need never compute BFS beyond distance $D^* \geq w\beta D$ in $G^*$.
The $Z$-sequence is defined as follows.
\begin{align*}
Y[i] &= \max_{j\ge 0} \{2^j \mbox{ such that } 2^j|i\} & (i \ge 1)\\
\mbox{I.e., } Y &= (1,2,1,4,1,2,1,8,1,2,1,4,1,2,1,16,  1,2,1,4,1,2,1,8,1, 2,1,4,1,2,1,32, \ldots)\\
Z[0] &= D^\ast\\
Z[i]  &= \min\{ D^\ast, \ \alpha \cdot Y[i]\}, \quad \mbox{ where $\alpha = 4$} & (i \ge 1)\\
D^\ast &= \min_{j\ge 0} \{\alpha 2^j \mbox{ such that } \alpha 2^j \ge w\beta D\}
\end{align*}
In other words, $Z$ is derived by multiplying $Y$ by $\alpha = 4$, 
truncating large elements at $D^*$, 
and beginning the sequence at 0, with $Z[0]=D^*$.
(Here $Z[0]$ corresponds to the distance threshold $D^*$ used in 
Step 1 of $\bfs$ to estimate distances to the $0$th wavefront $W_0 = S$.)

%\input{figure-time-evolution}

%%%%%%%%%%%%%%%%%%%%%%%%%%%%%%%%%%%%%%

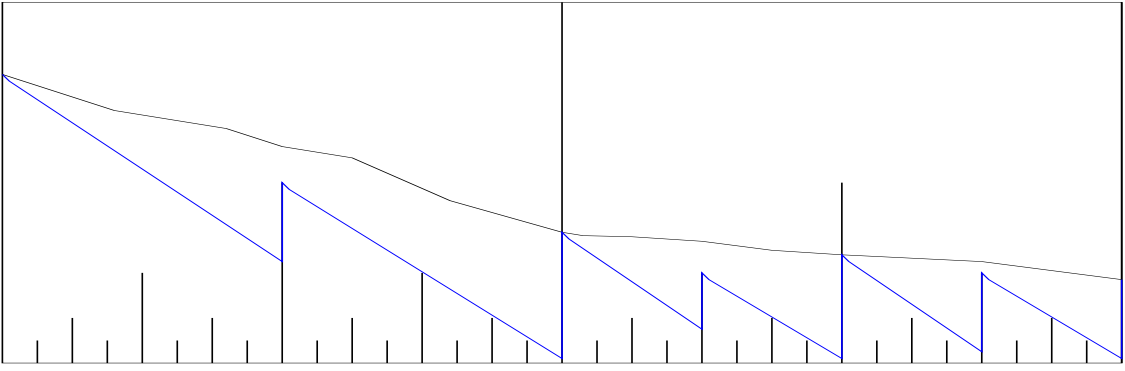
\begin{figure*} 
\begin{center}
\scalebox{1.55}[0.5]{
\begin{tikzpicture}[scale=0.6]
%\draw (0,0) --(1,2);
%\draw[help lines] (0,0) grid (2,3);
% Outer box:
\draw (0,0) rectangle (16,16);
% Tick marks:
\draw (1,0)--(1,2);
\draw (2,0)--(2,4);
\draw (3,0)--(3,2);
\draw (4,0)--(4,8);
\draw (5,0)--(5,2);
\draw (6,0)--(6,4);
\draw (7,0)--(7,2);
\draw (8,0)--(8,16);
\draw (9,0)--(9,2);
\draw (10,0)--(10,4);
\draw (11,0)--(11,2);
\draw (12,0)--(12,8);
\draw (13,0)--(13,2);
\draw (14,0)--(14,4);
\draw (15,0)--(15,2);
\draw (16,0)--(16,16);
\draw (0.5,0)--(0.5,1);
\draw (1.5,0)--(1.5,1);
\draw (2.5,0)--(2.5,1);
\draw (3.5,0)--(3.5,1);
\draw (4.5,0)--(4.5,1);
\draw (5.5,0)--(5.5,1);
\draw (6.5,0)--(6.5,1);
\draw (7.5,0)--(7.5,1);
\draw (8.5,0)--(8.5,1);
\draw (9.5,0)--(9.5,1);
\draw (10.5,0)--(10.5,1);
\draw (11.5,0)--(11.5,1);
\draw (12.5,0)--(12.5,1);
\draw (13.5,0)--(13.5,1);
\draw (14.5,0)--(14.5,1);
\draw (15.5,0)--(15.5,1);

% "real" curve
\draw (0,12.8)--(1.6,11.2)--(3.2,10.4)--(4,9.6)--(5,9.1)--(6.4,7.2)--(8,5.8)--(8.3,5.65)--(9,5.6)--(10,5.4)--(11,5)--(12,4.8)--(14,4.5)--(16,3.7);

% lower bound curve
\draw [blue] (0,12.8)--(0.1,12.5)--(4,4.5)--(4,8)--(4.1,7.7)--(8,0.2)
--(8,5.8)--(8.1,5.5)--(10,1.5)--(10,4)--(10.1,3.7)--(12,0.2)--(12,4.8)--(12.1,4.5)--(14,0.5)--(14,4)--(14.1,3.7)--(16,0.2)--(16,3.7);
\end{tikzpicture}
}
\end{center}
\caption{\label{fig:time-evolution}\small Part of the time evolution of the distance of a fixed cluster, $C$, from the
  frontier in the cluster graph $G^*$.  The $x$-axis is time spent moving the wavefront across
  the underlying graph $G$.  Every vertical tick mark is a time at
  which this is suspended, so that, recursively, BFS can be done on the
  cluster graph $G^*$, starting from the current wavefront.  
  The height of each such tick mark indicates the depth to which this search
  is to be done.  The $y$-axis is the distance of $C$ to
  the wave front, \underline{\emph{in $G^*$}}.  The top curve shows the
  irregular, but monotonic, decrease of this distance over time.
  The bottom curve, in blue, shows the high-probability lower bound on this distance, 
  from the perspective of the cluster in question.  Note that every
  time the top curve intersects a tick mark, the cluster must
  participate in the BFS on the cluster graph, or this BFS will fail.  
  Every time the bottom
  curve intersects a tick mark, the cluster will wake up in order to participate
  in the BFS, because it thinks it \emph{may} be needed.  Note that, by design,
  the lower curve often passes just above the tick marks without
  actually intersecting them.   The reader should bear in mind that these
  two curves chart the actual/likely distance of $C$ to the wavefront in $G$; 
  the algorithm maintains the related interval $[L_i(C),U_i(C)]$, which bounds
  distances from $C$ to the wavefront \underline{\emph{in $G$}}.
}
\end{figure*}

%%%%%%%%%%%%%%%%%%%%%%%%%%%%%%%%%%%%%%

Figure~\ref{fig:time-evolution} gives an example, from the perspective of a single cluster, of how the distance estimate evolve over time.  

\paragraph{Organization of Section~\ref{sect:BFS}.}
In Section~\ref{sect:BFSlemmas} we prove a number of lemmas
that relate to the correctness and efficiency
of $\bfs$, including proofs of 
Claims~\ref{claim:X} and \ref{claim:U}.
In Section~\ref{sect:BFSmainthm} we analyze the overall 
time and energy-efficiency of the BFS algorithm.

\subsection{Auxiliary Lemmas}\label{sect:BFSlemmas}

Lemma~\ref{lem:distance-estimates} justifies how distance estimates are 
updated in Steps 1, 7, and 8 of $\bfs$ in order to preserve Invariant~\ref{inv:interval}, with high probability.

\begin{lemma}\label{lem:distance-estimates}
Let $W_i$ be the $i$th wavefront; 
let $\Upsilon$ include all clusters $C$ such that 
$\dist_G(W_i,C) \in [i\beta^{-1}, (i+Z')\beta^{-1}]$;
and let $G_i^*$ be the subgraph of $G^*$ induced by $\Upsilon$.
If $\cluster(u) \in \Upsilon$ and $\dist_G(S,u) \ge i\beta^{-1}$,
then w.h.p.,
\begin{dmath*}
\dist_G(W_i,u) \in \left[\min\left\{\frac{Z'}{\beta}+1,\; \dist_{G_i^*}(W_i^*,\cluster(u))\cdot \frac{1}{w\beta}\right\}, 
\max\left\{1,\, \dist_{G_i^*}(W_i^*,\cluster(u))\right\}\cdot \frac{w}{\beta}
\right].
\end{dmath*}
%\[
%\dist_G(W_i,u) \in \left[\min\left\{\frac{Z'}{\beta}+1,\; \dist_{G_i^*}(W_i^*,\cluster(u))\cdot \frac{1}{w\beta}\right\},\;
%\max\left\{1,\, \dist_{G_i^*}(W_i^*,\cluster(u))\right\}\cdot \frac{w}{\beta}
%\right]
%\]
\end{lemma}

\begin{proof}
If $d = \dist_G(W_i,u) \ge Z'\beta^{-1}+1$ then the lower bound is already correct, so suppose that $d \le Z'\beta^{-1}$.  
Let $P$ be any length-$d$ path from $u$ to $W_i$ in $G$. Lemma~\ref{lem:dist-ratio-union-bound} implies that w.h.p., there is a path
$P^*$ in $G_i^*$ from $W_i^*$ to $\cluster(u)$ 
with length at most $O(\beta d \log n) < w \beta d$, 
and so $\dist_{G_i^*}(W_i^*,\cluster(u)) \leq w \beta d$, as required.

This upper bound follows from the cluster diameter upper bound 
$K = 8\log(n)/\beta \leq w/(2\beta)-1$.
Thus, if $\dist_{G_i^*}(W_i^*, \cluster(u)) = d'$
then $\dist_G(W_i,u) \le (d'+1)\cdot(K+1) \leq \max\{d'+1\}\cdot w\beta^{-1}$.
\end{proof}

Lemma~\ref{lem:distance-estimates} shows that Step 1 of $\bfs$ 
initializes $L_0(\cdot), U_0(\cdot)$ to satisfy Invariant~\ref{inv:interval},
w.h.p.  Here $\Upsilon = A^*$
is the set of all active clusters; if $\dist_A(S,u) \in [0,D]$ 
(the relevant range), then Lemma~\ref{lem:distance-estimates} 
guarantees that 
$\dist_A(S,u) \in [L_0(\cluster(u)),U_0(\cluster(u))]$ after Step 1.
The estimates set in Step 8 of $\bfs$ are trivially correct;
Lemma~\ref{lem:distance-estimates} also guarantees that the lower
and upper bounds fixed in Step 7 are correct.

We use several properties 
of the $Z$ sequence, listed in Lemma~\ref{lem:Z}.
\begin{lemma}\label{lem:Z}
Fix an index $i$.
\begin{enumerate}
    \item For any number $b \geq \alpha$,
    define $j>i$ to be the smallest index such that $Z[j] \geq b$.
    Then 
    \[j-i \leq  b/\alpha.\] 
    Suppose the number $b$ additionally satisfies that $b \leq Z[i]$ and $b \in \{\alpha, 2\alpha, 4\alpha, 8\alpha, \ldots D^\ast\}$. Then we have $Z[i] = b$ and $j-i = Z[j]/\alpha$. \label{Seq-property-1}
    
    \item Define $j > i$ to be the smallest index such that 
    $Z[j] > Z[i]$ or  $Z[j] = D^\ast$. 
    Then we have $j-i = Z[i]/\alpha$; moreover, all indices $k \in \{i+1, \ldots, j-1\}$ satisfy that $Z[k] \leq Z[i]/2$. \label{Seq-property-2}
\end{enumerate}
\end{lemma}

\begin{proof} %(sketch)
Parts 1 and 2 follow 
from the fact that in the $Y$-sequence, 
the values at least $2^\ell$ appear periodically
with period $2^\ell$.  Thus, the 
values at least $\alpha 2^\ell$ 
in the $Z$-sequence also appear periodically with
period $2^\ell$.
\end{proof}

We are now prepared to prove Claim~\ref{claim:X}.
\begin{proof}[Proof of Claim~\ref{claim:X}]
It follows from Invariant~\ref{inv:interval} 
that $X_i$, as defined in Step 4 of $\bfs$, 
includes all active vertices within distance $\beta^{-1}$
of the $i$th wavefront $W_i$.  It remains to show
no $u$ is included in $X_i$ for more than $\poly(\log n)$ indices $i$.

Suppose that $u \in X_i$ for $i>0$.  It follows that
$L_i(\cluster(u)) \le \beta^{-1}$ and that in the previous stage,
$L_{i-1}(\cluster(u)) \leq 2\beta^{-1}$.
Since $Z[i] \ge \alpha = 4$, it must have been 
that $\cluster(u)$ was included in $\Upsilon$ and participated
in the Special Update (Step 7 of $\bfs$) before stage $i$.
If $\dist_{G_i^*}(W_i^*,\cluster(u))=x$ and after the Special Update,
$L_i(\cluster(u)) \le \beta^{-1}$, it must be that $x\leq w$,
and hence $U_i(\cluster(u)) \le w^2 \beta^{-1}$.  Thus, $u$
may participate in at most $w^2$ more stages 
(joining $X_i,X_{i+1},\ldots,X_{i+w^2}$)
before its distance is settled and it is 
deactivated, in Step 6 of $\bfs$.
\end{proof}

Before proving Claim~\ref{claim:U} we begin with
three auxiliary lemmas, Lemmas~\ref{lem-aux-1-ssss},
\ref{lem-aux-11-ssss}, and \ref{lem-aux-2-ssss}.

\begin{lemma}\label{lem-aux-1-ssss}
Recall $\alpha = 4$.
Suppose cluster $C$ is included in $G_i^*$ and $G_j^*$, 
but not in $G_{i'}^*$ for any $i'\in \{i+1, \ldots, j-1\}$. 
Then  we have
\[
\frac{L_i(C)}{8\alpha} \leq \frac{j-i}{\beta} \leq \max\left\{\frac{1}{\beta},\;  \frac{L_i(C)}{\alpha}\right\}.
\]
\end{lemma}

\begin{proof}
We prove the upper and lower bounds on  $(j-i)/\beta$ separately.

\paragraph{Upper Bound.}
%First, we prove that $(j-i)/\beta \leq \max\{1/\beta,  L_i(C)/\alpha\}$. 
Select $j^* > i$ to be the first stage index for which
$Z[j^\ast] \geq \min\{ D^\ast, \beta L_i(C)\}$. 
Clearly $j^* \geq j$ since if $L_{i+1}(C),\ldots,L_{j^*}(C)$ were set
according to Automatic Updates we would have 
$L_{j^*-1}(C) \leq L_i(C) \leq Z[j^*] \beta^{-1}$, 
which would trigger a Special Update to $L_{j^*}(C)$. 
There are two cases to consider, 
either of which establishes the upper bound
on $(j-i)/\beta$.
\begin{itemize}
    \item Suppose $\beta\cdot L_{i}(C) < \alpha$. Then $j = j^\ast = i+1$, and so $(j-i)/\beta = 1/\beta$.
    
    \item Suppose $\beta\cdot L_{i}(C) \geq \alpha$. 
    % Then $\alpha \leq \min\{ D^\ast, \beta L_i(C)\} <   \beta L_i(C)$. 
    According to Lemma~\ref{lem:Z}(\ref{Seq-property-1}),
    $j^\ast - i \leq  \min\{ D^\ast, \beta L_i(C)\}/\alpha$, and so 
    $(j-i)/\beta
    \leq (j^\ast -i)/\beta \leq \min\{ D^\ast, \beta L_i(C)\}/(\alpha \beta) \leq  L_i(C)/\alpha$.
\end{itemize}

\paragraph{Lower Bound.}
In order to prove that $(j-i)\cdot \beta^{-1} \geq L_i(C)/(8\alpha)$
it suffices to find any particular index $j^*$ such that:
\begin{enumerate}
    \item $(j^\ast-i)\cdot \beta^{-1} \geq L_i(C)/ (8 \alpha)$.
    \item For all $j' \in [i+1, j^\ast]$, $C$ is not included in $G_{j'}^*$.
\end{enumerate}
Condition 2 implies that $j^\star < j$ and then Condition 1 implies
that $(j-i)/\beta > (j^\ast-i)\cdot \beta^{-1} \geq L_i(C)/ (8 \alpha)$, as desired.
We will explain how to select $j^*$ shortly.  In the meantime, consider
the following two conditions;
we will argue that (a) and (b) imply Condition 2 above. 

%To prove Condition~2, it suffices to show that $L_{j'-1}(C) > (Z[j']+1)/\beta$ for all  $j' \in [i+1, j^\ast]$, in view of the criterion for $C \in \Upsilon$. The following two conditions will together imply Condition~2.
\begin{enumerate}
    \item[(a)] For all $j' \in [i+1, j^\ast - 1]$, 
    we have $Z[j'] < Z[j^\ast]$.
    \item[(b)] $L_{i}(C) - (j^\ast - i)\cdot \beta^{-1} > Z[j^\ast]\cdot \beta^{-1}$.
\end{enumerate}

Recall that $C$ is \emph{not} included in $G_{j'}^*$
iff $L_{j'-1}(C) > (Z[j']+1)\cdot \beta^{-1}$, so it
suffices to prove the latter inequality for every
$j' \in [i+1,j^*]$.  By induction, we can 
assume that the claim is true for all $j'' \in [i+1,j'-1]$,
i.e., $L_{j''}(C)$ was set according to an Automatic Update
(Step 8) and $L_{j''}(C) = L_i(C) - (j''-i)\cdot \beta^{-1}$.  Thus,
\begin{align*}
    L_{j'-1}(C) 
    &= L_{i}(C) - ((j'-1) - i)\cdot \beta^{-1}
    &&&& \text{Follows from induction hypothesis}\\
    &\geq L_{i}(C) - ((j^\ast-1) - i)\cdot \beta^{-1}\\
    & > (Z[j^\ast]+1)\cdot \beta^{-1} &&&& \text{by (b)}\\
    & >  (Z[j']+1)\cdot \beta^{-1} &&&& \text{by (a)}
\end{align*}

\paragraph{Choice of $j^\ast$.}
Select $x$ to be the integer in
$\{ \alpha, 2\alpha, 4\alpha, 8\alpha, \ldots, D^\ast\}$
such that 
\[
x \in \left[\frac{\beta\cdot L_i(C)}{8},\;  \frac{\beta\cdot L_i(C)}{4}\right).
\]
It is guaranteed that $x$ exists so long as 
$\beta\cdot L_i(C) > 4\alpha$.  When $\beta\cdot L_i(C) \leq 4\alpha$, we already have the desired lower bound on $(j-1)\cdot\beta^{-1}$ since
$L_i(C)/(8\alpha) \leq \beta^{-1}/2 < \beta^{-1} \leq (j-i)\cdot\beta^{-1}$.

Observe that  $Z[i]$, like $x$, is also an integer in $\{ \alpha, 2\alpha, 4\alpha, 8\alpha, \ldots, D^\ast\}$. 
In a Special Update, the largest value
that $L_i(C)$ can attain is $Z[i]\cdot\beta^{-1}+1$,
hence
\[
Z[i] \geq \beta\cdot (L_i(C)-1) > \beta\cdot L_i(C)/2 > 2x, 
\]
Define $j^\ast > i$ to be the smallest index such that 
$Z[j^\ast] \geq x$. 
In particular, since $Z[i] \geq 2x > x$,   Lemma~\ref{lem:Z}(\ref{Seq-property-1}) guarantees that
$Z[j^\ast] = x$ and hence
\[
j^\ast-i  = Z[j^\ast]/\alpha = x/\alpha \geq \beta\cdot L_i(C)/ (8 \alpha).
\]
Thus Condition~1 is met for this choice of $j^*$.

Condition~(a) is also met, since by definition of $j^*$,
$Z[j'] < x = Z[j^*]$ for all $j' \in [i+1,j^*-1]$.
Now we turn to Condition~(b).  
Observe that 
\begin{equation}\label{eqn:ub}
j^\ast-i  = Z[j^\ast]/\alpha = x/\alpha < \beta\cdot  L_i(C)/ (4 \alpha).
\end{equation}
We prove that $L_{i}(C) - (j^\ast - i)\cdot \beta^{-1} > Z[j^\ast] \cdot \beta^{-1}$.
\begin{align*}
    Z[j^\ast]\cdot \beta^{-1}
    &< 2x\cdot \beta^{-1} &&&& \text{ since } Z[j^\ast] < 2x\\
    &< L_i(C)/2 &&&&  \text{ since } x \in [\beta\cdot L_i(C)/8,  \beta\cdot L_i(C)/4)\\
    &= L_i(C)(1 - 2/\alpha) &&&&  \text{ since } \alpha = 4\\
    &< L_i(C) - 8(j^\ast - i)\cdot \beta^{-1} &&&&  \text{ by (\ref{eqn:ub}), } (j^\ast-i)\cdot\beta^{-1} < L_i(C)/(4 \alpha)\\
    &< L_i(C) - (j^\ast - i)\cdot \beta^{-1}
\end{align*}
Conditions (a) and (b) imply Condition 2, 
which implies $L_i(C)/(8\alpha)  \leq (j-i)\cdot \beta^{-1}$.
\end{proof}

\begin{lemma}\label{lem-aux-11-ssss}
Suppose $C$ appears in $G_i^*$ and $G_j^*$ 
but not in $G_{i'}^*$ for any $i' \in \{i+1, \ldots, j-1\}$. 
Suppose that when $L_i(C)$ is set during a Special Update (Step 7 of $\bfs$), we have $L_{i}(C) =(Z[i]/\beta)+1$. 
It must be that $Z[j] > Z[i]$ or $Z[j] = D^\ast$.
\end{lemma}

\begin{proof}
Define $j^\ast > i$ to be the \emph{smallest} index such that 
$Z[j^\ast] > Z[i]$ or  $Z[j^\ast] = D^\ast$. 
To prove the lemma it suffices to show that $j^* = j$, i.e.,
$L_{j'}(C)$ is set according to an Automatic Update
for $j' \in \{i+1,\ldots,j^*-1\}$ but $C$ appears in $G_{j^*}$
and participates in a Special Update.

To prove that $L_{j'}(C)$ is set according to an Automatic Update
(assuming, inductively, that the claim holds for $L_{i+1}(C),\ldots,L_{j'-1}(C)$)
it suffices to show 
\[
L_{j'-1}(C) - \beta^{-1} 
= L_i(C) - (j'-i)\cdot\beta^{-1} 
> Z[k]\cdot \beta^{-1}.
\]
By Lemma~\ref{lem:Z}(\ref{Seq-property-2}), 
$j^*-i = Z[i]/\alpha$.  Since $j'<j^*$ we have
\[
(j'-i)\cdot \beta^{-1} 
< (j^\ast - i)\cdot \beta^{-1} 
= Z[i]\cdot \beta^{-1}/\alpha 
<  L_{i}(C) / \alpha.
\]
It follows that 
\[
L_i(C) - (j'-i)\cdot\beta^{-1}
> (1-1/\alpha)L_i(C) > L_i(C)/2.
\]
On the other hand, Lemma~\ref{lem:Z}(\ref{Seq-property-2}) implies that 
\[
Z[j']\cdot\beta^{-1}
\leq 
(Z[i]/2)\cdot\beta^{-1}
< L_i(C)/2.
\]
Therefore $L_i(C) - (j'-i)\beta^{-1} > Z[j']\cdot\beta^{-1}$,
implying $L_{j'}(C)$ is set according to an Automatic Update.
Finally, from the definition of $i$ and $j^*$ we have
\[
L_{j^*-1}(C) < L_i(C)
= Z[i]\cdot\beta^{-1}+1
\leq Z[j^*]\cdot \beta^{-1}+1
< (Z[j^*]+1)\cdot \beta^{-1},
\]
meaning $C$ appears in $G_{j^*}^*$ and $L_{j^*}(C)$ is set according to a Special Update.
\end{proof}

In the $\bfs$ algorithm, 
the upper bound estimates $U_i(C)$ are 
all monotonically decreasing
with $i$, due to the way Special and Automatic Updates are performed in Steps 7 and 8.  On the other hand, the 
lower bound estimates $L_i(C)$ are only monotonically decreasing during Automatic Updates and may oscillate many times over the
execution of the algorithm.  
(See Figure~\ref{fig:time-evolution} for a depiction of how this happens.)
Since $U_{\cdot}(\cdot)$-values
offer a more stable way to measure progress, we need
to connect them with the $L_{\cdot}(\cdot)$-values, 
which directly influence the composition of $X_i$ and $G_i^*$.

\begin{lemma}\label{lem-aux-2-ssss}
If $[L_i(C),U_i(C)]$ is set during a Special Update step, then
\[
U_i(C) \leq \max\{2w^2\cdot L_i(C),\; 2w^2\cdot \beta^{-1}\}
\]
\end{lemma}

\begin{proof}
The proof is by induction on $i$.
We regard Step 1 of $\bfs$ as the Special Update for $i=0$.
The claim clearly holds for $i=0$ since $U_0(C)$ is set such that
$U_0(C) \in \{w\beta^{-1}, w^2\cdot L_0(C)\}$.
Assume, inductively, that the lemma holds for all indices less than $i$.

In general, whenever $L_i(C)$ is set to be $x\beta^{-1}/w$ in 
Step 7, where $x= \dist_{G_i^*}(W_i^*,C)$, the claim holds since
$U_i(C) \in \{w\beta^{-1}, w^2\cdot L_i(C)\}$.
Thus, we may proceed under the assumption that
$L_i(C)$ is set to be $Z[i]\cdot\beta^{-1}+1$ during a Special Update.

Define $i^\ast < i$ to be the last stage
in which $L_{i^*}(C)$ was set by a Special Update.
We consider two cases, depending on how $L_{i^*}(C)$ was set.
\begin{itemize}
    \item Suppose $L_{i^\ast}(C)$ is set to be 
    $x/(\beta w) \leq Z[i^\ast]/\beta$ in the Special Update,
    and as a consequence, 
    $U_{i^*}(C) \leq \max\{w\beta^{-1}, w^2 \cdot L_{i^*}(C)\}$.
    (Here $x > 0$ is the BFS-label of $C$ found in Step 7.)
    If $U_{i^\ast}(C) \leq (2w^2+1)\cdot \beta^{-1}$, 
    then we are already done, since 
    $U_i(C) \leq U_{i^\ast}(C) - 1/\beta \leq 2w^2\cdot \beta^{-1}$. 
    Thus, we may assume $U_{i^\ast}(C) > (2w^2+1)\cdot\beta^{-1}$,
    and consequently, that $L_{i^\ast}(C) > 2\beta^{-1}$.
     
    By Lemma~\ref{lem-aux-1-ssss}, we have $(i-{i^\ast})/\beta \leq \max\{1/\beta, L_{i^\ast}(C)/\alpha\} < L_{i^\ast}(C)/2$.
    In order for $L_i(C)$ to be set by a Special Update, 
    it is necessary that $L_{i-1}(C) \leq (Z[i]+1)\cdot \beta$.  
    Thus, we must have
\begin{align*}
    Z[i]\cdot \beta^{-1} &\geq L_{i-1}(C)-\beta^{-1} 
                &  & & & \text{since $L_{i-1}(C) \leq (Z[i]+1)\cdot \beta^{-1}$}\\
    &= L_{i^\ast}(C) - (i-{i^\ast})\cdot\beta^{-1} 
                & & & & \text{since $C$ does not appear in $G_{i^\ast+1}^*,\ldots,G_{i-1}^*$}\\
    &\geq L_{i^\ast}(C)/2 & & & & \text{since  $(i-{i^\ast})\cdot \beta^{-1} \leq L_{i^\ast}(C)/2$}
\end{align*}
%Hence $Z[i]/\beta \geq L_{i^\ast}(C)/2$.
    Remember that $L_i(C) = Z[i]\cdot \beta^{-1}+1$, and based on this we show that $U_i(C) \leq 2w^2 \cdot L_i(C)$.
\begin{align*}
    L_i(C) &= Z[i]\cdot \beta^{-1}+1 \\
    &>   L_{i^\ast}(C)/2 & & &  &\text{ since } Z[i]\cdot\beta^{-1} \geq L_{i^\ast}(C)/2\\
    &  \geq U_{i^\ast}(C) / (2w^2) & & & &\text{ since $U_{i^\ast}(C) \leq w^2  L_{i^\ast}(C)$}\\
    &> U_i(C) / (2w^2) & & & &\text{ since $U_i(C) < U_{i^\ast}(C)$, as $i^\ast < i$.}
\end{align*}    
\item Now consider the case when $L_{{i^\ast}}(C)$ is set to be $Z[{i^\ast}]\cdot\beta^{-1}+1$. By Lemma~\ref{lem-aux-11-ssss}, we have $Z[i] \geq Z[i^\ast]$.
Therefore, 
$L_i(C) = Z[i]\cdot\beta^{-1}+1  \geq Z[i^\ast]\cdot\beta^{-1}+1 = L_{i^\ast}(C)$.
By the inductive hypothesis,
it is guaranteed that 
$U_{i^\ast}(C) \leq \max\{2w^2\cdot\beta^{-1},\, 2w^2\cdot L_{i^\ast}(C)\}$.
If $U_{i^\ast}(C) \leq 2w^2\cdot \beta^{-1}$, then we are done. 
If $U_{i^\ast}(C) \leq 2w^2\cdot L_{i^\ast}(C)$, then we have 
\[
L_i(C) \geq L_{i^\ast}(C) \geq U_{i^\ast}(C) / (2w^2) > U_i(C) / (2w^2).
\]
\end{itemize}
This concludes the induction and the proof.
\end{proof}

We are now in a position to prove Claim~\ref{claim:U}, that each
vertex participates in $G_i^*$ for at most 
$\tilde{O}(1)$ indices $i$.

\begin{proof}[Proof of Claim~\ref{claim:U}]
Suppose that $C$ participates in a Special Update that 
sets $[L_i(C),U_i(C)]$ with $U_i(C) \ge 2w^2\cdot \beta^{-1}$
and that the next interval to be set by a Special Update is $[L_{j}(C),U_{j}(C)]$.
Then
\begin{align}
(j-i) 
&\geq \frac{\beta\cdot L_i(C)}{8\alpha}
\geq \frac{\beta\cdot U_i(C)}{16\alpha w^2}.\label{eqn:blah}
\end{align}
The first inequality of (\ref{eqn:blah})
follows from Lemma~\ref{lem-aux-1-ssss}
and the second inequality from Lemma~\ref{lem-aux-2-ssss}.  
Since $U_*(C)$ is decremented
by at least $\beta^{-1}$ in each stage,
(\ref{eqn:blah}) implies that
\[
U_{j}(C) \leq U_i(C) - (j-i)\cdot\beta^{-1} \leq U_i(C)\left(1-\frac{1}{16\alpha w^2}\right).
\]
In other words, $C$ participates in at most
$\log_{1+\Theta(1/w^2)} D = \Theta(w^2 \log D) = O(\log^3 n)$ 
Special Updates until some stage $i$ in which
$U_i(C) < 2w^2\cdot \beta^{-1}$, after which $C$
participates in at most $O(w^2)$ Special Updates
all constituents of $C$ settle their distance
from the source and are deactivated.
\end{proof}

\subsection{Time and Energy Complexity of BFS}\label{sect:BFSmainthm}

The remainder of this section constitutes a proof of
Theorem~\ref{thm:main}.

\begin{theorem}\label{thm:main}
Let $G=(V,E)$ be a radio network, 
$s\in V$ be a distinguished source vertex,
and $D = \max_u \dist_G(s,u)$.
A Breadth First Search labeling can be computed
in $\tilde{O}(D)\cdot 2^{O(\sqrt{\log D\log\log n})}$ time
and $\tilde{O}(1)\cdot 2^{O(\sqrt{\log D\log\log n})}$ energy,
with high probability.
\end{theorem}

The main problem is to compute BFS up to some threshold distance $D_0$.
Once we have a solution to this problem, we can obtain bounds
in terms of the (unknown) $D$ parameter by testing every 
$D_0 = 2^k$
that is a power of 2, stopping at the first value that labels
all of $V(G)$.  We use a call to $\sr$ as a unit of measurement of both time and energy, i.e., calling $\sr$ takes one unit of time, and every \emph{participating} vertex expends one unit of energy.
(By Lemma~\ref{lemma:sr-decay} actual time and energy are at most a 
$O(\log^2 n)$ factor larger.)

The algorithm we apply is a slightly modified $\bfs$, 
where all cluster graphs in all recursive invocations are constructed with $\beta = 2^{-\sqrt{\log D_0 \log\log n}}$.
We only apply $\bfs$ to recursion depth 
$L = \sqrt{\log D_0/\log\log n}$, at which point we revert 
to the trivial BFS algorithm that settles all distances up to $D'$
using $D'$ time and energy, by calling $\sr$ $D'$ times.

Define $\Energy_r(D')$ to be the number of calls to $\sr$ that a 
vertex participates in when computing BFS to distance $D'$, 
and when the recursion depth is $r \in [0,L]$.
Thus, we have
\[
\Energy_r(D') = \left\{
\begin{array}{lr}
\tilde{O}(1)\cdot \Energy_{r+1}(\tilde{O}(\beta D')) + \tilde{O}(\beta^{-1}) & \mbox{if $r < L$}\\
D'                      & \mbox{if $r=L$}
\end{array}\right.
\]
By Lemma~\ref{lem:clustr-diam-ub} the cost to create the cluster
graph $G^*$ is $\tilde{O}(\beta^{-1})$.
By Claim~\ref{claim:X} each vertex appears in $X_i$
for $\tilde{O}(1)$ stages $i$, and for each, participates
in $\beta^{-1}$ calls to $\sr$.  These costs are covered
by the $\tilde{O}(\beta^{-1})$ term.
All calls to $\bfs$ on $G^*$ involve computing BFS to 
some distance at most $D^* = w\beta D' = \tilde{O}(\beta D')$.
By Claim~\ref{claim:U}, every vertex participates in 
$\tilde{O}(1)$ such recursive calls.
Moreover, by Lemma~\ref{lem:sr-sim}, 
every cluster $C$ (vertex in $G^*$) that 
participates in a call to $\sr$ on $G^*$ can be simulated 
such that constituent vertices of $C$ participate in 
$\tilde{O}(1)$ calls to $\sr$ on $G$.  The costs of recursive
calls are represented by the $\tilde{O}(1)\cdot \Energy_{r+1}(\tilde{O}(\beta D'))$ term.

When the recursion depth $r$ reaches $L$, 
the \emph{maximum} value
of $D'$ is therefore at most
\[
D_L = D_0 \cdot (\tilde{O}(\beta))^L = (\tilde{O}(1))^L = 2^{O(\sqrt{\log D_0\log\log n})},
\]
since $\beta^L = D_0^{-1}$.  Thus, the energy cost of the top-level recursive call is at most
\[
\Energy_0(D_0) = (\tilde{O}(1))^L \cdot (D_L + \tilde{O}(\beta^{-1})) = \tilde{O}(1)\cdot 2^{O(\sqrt{\log D_0 \log\log n})}.
\]

We can set up a similar recursive expression 
for the time of this algorithm.
\[
    \Time_r(D') \le \left\{
    \begin{array}{lr}
    \displaystyle O(D') + \tilde{O}(\beta^{-1})\cdot 
    \sum_{i=0}^{\ceil{\beta D'}-1} \Time_{r+1}(Z[i]) & \mbox{if $r<L$}\\
    D'  & \mbox{ if $r=L$}
    \end{array}
    \right.
\]
The $r=L$ case is the time of the trivial algorithm,
so we focus on justifying the expression for $r<L$.
The time to advance the BFS wavefront over all $\ceil{\beta D'}$
stages of Step 5 is $O(D')$.  
We treat Step 1 as the Special Update for $i=0$ with $Z[0]=D^*$.
In general, the Special Update for stage $i$ takes
$\Time_{r+1}(Z[i])$ time \underline{with respect to $G^*$},
and each unit of time (i.e., a call to $\sr$) is simulated
in $G$ in time linear in the maximum cluster diameter, 
namely $\tilde{O}(\beta^{-1})$.  By Lemma~\ref{lem:Z},
each value $b\in B = \{\alpha,2\alpha,4\alpha,\ldots,D^*\}$
appears less than $(\beta D'/b)$ times in 
$Z[0],\ldots,Z[\ceil{\beta D'}-1]$, 
hence we can rewrite the sum as 
$\sum_{b\in B} (\beta D'/b)\cdot \Time_{r+1}(b)$.
Assuming inductively that 
$\Time_{r+1}(b)$ is 
$b\cdot(\tilde{O}(1))^{L - (r+1)}$, which holds when $r+1=L$,
we have
\begin{align*}
    \Time_r(D') &\leq O(D') + \tilde{O}(\beta^{-1})\cdot 
    \sum_{b\in B} (\beta D'/b)\cdot \Time_{r+1}(b)\\
    &= O(D') + \tilde{O}(1)\cdot \sum_{b\in B} (D'/b)\cdot b\cdot (\tilde{O}(1))^{L-(r+1)}\\
    &= D' \cdot (\tilde{O}(1))^{L-r}
\end{align*}
Hence $\Time_0(D_0) = D_0 \cdot (\tilde{O}(1))^L = \tilde{O}(D_0)\cdot 2^{O(\sqrt{\log D_0\log\log n})}$.

%
%
%
%
%
%
%
%
%
%

%\section{Hardness of Diameter Approximation}\label{sect:diameter}

\section{Hardness of Diameter Approximation}\label{appendix:diameter}

In this section, we show that certain approximations of diameter cannot be computed in $o(n)$ energy, even allowing messages of unlimited size.
Our lower bounds also hold in the setting where the network supports {\em collision detection}, i.e., in each time slot $t$, each listener $v$ is able to distinguish between the following two cases: (i) at least two vertices in $N(v)$ transmit at time $t$ (noise), or  (ii) no vertex in   $N(v)$ transmits at time $t$ (silence). 

First, we show that computing a $(2-\epsilon)$-approximation of diameter is hard by proving that it takes $\Omega(n)$ energy to distinguish between  (i) an $n$-vertex complete graph $K_n$ (which has diameter $1$), or (ii) an $n$-vertex complete graph minus one edge $K_n - e$ (which has diameter 2).

\begin{theorem}\label{thm:diameter-lb1}
The energy complexity of computing a $(2-\epsilon)$-approximation of diameter is $\Omega(n)$, even on the class of unit-disc graphs.
\end{theorem}
\begin{proof}
Throughout the proof, we consider the scenario where the underlying graph is $K_n$ with probability $1/2$, and is $K_n - e$  with probability $1/2$. The edge $e$ is chosen uniformly at random.  Observe that
both $K_n$ and $K_n - e$ are both unit disc graphs.
Let $\mathcal{A}$ be any randomized algorithm that   is able to  distinguish between $K_n$  and $K_n - e$.
We make the following simplifying assumptions, 
which only increase the capabilities of the vertices.
\begin{itemize}
\item Each vertex has a distinct ID from $[n]$.
\item All vertices have access to a shared random string.
\item By the end of each time slot $t$, each vertex knows the following information: (i) the IDs of the vertices transmitting at time $t$, (ii) the IDs of the vertices listening at time $t$, and (iii) the channel feedback (i.e., noise, silence, or a message $m$) for each listening vertex.
\end{itemize}

With the above extra capabilities, all vertices share the same history.
Since the actions of the vertices at time $t+1$ depend only on the shared history of all vertices and their shared random bits, by the end of time $t$ all vertices are able to predict the actions (i.e., transmit a message $m$, listen, or idle) of all vertices at time $t+1$.

We say that time $t$ is {\em good} for  a pair $\{u,v\}$ if the following conditions are met. Intuitively, if $t$ is not good for $\{u,v\}$, then what happens at time $t$ does not reveal any information as to whether $\{u,v\}$ is an edge.
\begin{itemize}
\item The number of transmitting vertices at time $t$ is  either 1 or 2,
\item One of the two vertices  $\{u,v\}$ listens at time $t$, and the other one transmits at time $t$.
\end{itemize}

Once the shared random string is fixed,
define $X_{\text{bad}}$ to be the set of pairs $\{u,v\}$ 
such that there is no time $t$ that is good for $\{u,v\}$ 
in an execution of $\mathcal{A}$ on $K_n$.
Define $X_{\text{good}}$ to be the remaining pairs.

We claim that if the energy per vertex is at most $E = (n-1)/8$, then for every pair $\{u,v\}$, $\Prob{\{u,v\} \in X_{\text{bad}}} \geq 1/2$.
Recall that if a time $t$ is good for some pair, then the number of transmitting vertices is at most 2.
Thus, if $t$ is good for $x$ pairs, then at least $x/2$ vertices
listen at time $t$, and so the total energy spent over all vertices and all time slots is at least $|X_{\text{good}}| / 2$.
On the other hand, it is also at most $nE = n(n-1)/8$.
If $n(n-1)/8 \geq |X_{\text{good}}| / 2$,
then $|X_{\text{bad}}| \geq  n(n-1)/4$
and
$\Prob{\{u,v\} \in X_{\text{bad}}} \geq  1/2$.

Recall that we pick $e$ at random and then choose 
the input graph to be either $K_n$ or $K_n-e$.
Once $e$ is selected, let
$\mathcal{E}$ be the event that $e\in X_{\text{bad}}$,
which now depends only on the shared random string.
When $\mathcal{E}$ occurs, the  execution of $\mathcal{A}$ is 
identical on both $K_n$ and $K_n - e$, 
and so the success probability of $\mathcal{A}$ is at most $1/2$.
Thus, $\mathcal{A}$  fails with probability at least $(1/2) \Prob{\mathcal{E}} \geq 1/4$.
This contradicts the assumption that $\mathcal{A}$  is able to  distinguish between $K_n$  and $K_n - e$.
\end{proof}

For sparse graphs (i.e., those with $O(\log n)$-arboricity), we show that $(3/2-\epsilon)$-approximation of diameter is hard.
The proof follows the framework of~\cite{AbboudCK16}, which shows that computing diameter takes $\Omega(n / \log^2 n)$ time in the $\CONGEST$ model,
or more generally $\Omega\left(\frac{n}{B \log n}\right)$ time in the message-passing model with $B$-bit message size constraint.
Note that a time lower bound in $\CONGEST$ does not, in general, imply any lower bound in $\RN[\infty]$, which \emph{has no message size constraint}. The main challenge for proving Theorem~\ref{thm:diameter-lb2} is that we allow  messages of unbounded length.

\begin{theorem}\label{thm:diameter-lb2}
The energy complexity of computing an $(3/2-\epsilon)$-approximation of diameter is $\Omega(n  / \log^2 n)$, even on graphs of $O(\log n)$-arboricity or $O(\log n)$ treewidth.
\end{theorem}
\begin{proof}
The proof is based on a reduction from the {\em set-disjointness} problem of communication complexity, which is defined as follows.
Consider two players $A$ and $B$, each of them holding a subset of $\{0, \ldots, n-1\}$. Their task is to decide whether their subsets are disjoint. If the maximum allowed failure probability is $f < 1/2$, then they need to communicate $\Omega(n)$ bits~\cite{BravermanM13,KalyanasundaramS92}.
This is true even if the two players have access to 
a public random string.

\paragraph{Lower Bound Graph Construction.}
Let $S_A=\{a_1, \ldots, a_{\alpha}\}$ and $S_B=\{b_1, \ldots, b_{\beta}\}$ be two subsets of $\{0, \ldots, k-1\}$ corresponding to an instance of set-disjointness problem. We assume that $k = 2^{\ell}$, for some positive integer $\ell$, and so each element $s \in S_A \cup S_B$ is represented as a binary string of length $\ell = \log k$. We write $\Ones(s) \subseteq [\ell] = \{1, \ldots, \ell\}$ to denote the set of indices $i$ in $[\ell]$ such that $s[i]=1$ (i.e., the $i$th bit of $s$ is $1$); similarly,  $\Zeros(s) = [\ell] \setminus \Ones(s)$ is the set of indices $i$ in $[\ell]$ such that $s[i]=0$. For example, if the binary representation of $s$ is $10110010$ ($\ell = 8$), then $\Ones(s) = \{1,3,4,7\}$ and $\Zeros(s) = \{2,5,6,8\}$.

Define the graph $G=(V,E)$ as follows. 
\begin{description}
\item[Vertex Set.] Define $V=V_A \cup V_B \cup V_C \cup V_D \cup \{u^\star, v^\star\}$, where $V_A = \{u_1, \ldots, u_\alpha\}$,
$V_B = \{v_1, \ldots, v_\beta\}$, $V_C = \{w_1, \ldots, w_{\ell}\}$, and $V_D = \{x_1, \ldots, x_{\ell}\}$. Note that we have natural 1-1 correspondences $V_A \leftrightarrow S_A$, $V_B \leftrightarrow S_B$,
$V_C \leftrightarrow [\ell]$, and $V_D \leftrightarrow [\ell]$.

\item[Edge Set.] The edge set $E$ is constructed as follows. Initially $E = \emptyset$.

For each vertex $u_i \in V_A$ and each $w_j \in V_C$, add $\{u_i, w_j\}$ to $E$ if $j \in \Ones(a_i)$.

For each vertex $u_i \in V_A$ and each $x_j \in V_D$, add $\{u_i, x_j\}$ to $E$ if $j \in \Zeros(a_i)$.

For each vertex $v_i \in V_B$ and each $w_j \in V_C$, add $\{v_i, w_j\}$ to $E$ if $j \in \Zeros(b_i)$.

For each vertex $v_i \in V_B$ and each $x_j \in V_D$, add $\{v_i, x_j\}$ to $E$ if $j \in \Ones(b_i)$.

Add edges between $u^\star$ and all vertices in $V_A \cup V_C \cup V_D$.

Add edges between $v^\star$ and all vertices in $V_B \cup V_C \cup V_D$.
\end{description}

The graph $G$ has $n=\alpha+\beta+2\ell+2 \le 2(k+\log k+1)$ vertices.
It is straightforward to show that $G$ has 
arboricity and treewidth $O(\log k) = O(\log n)$.

 A crucial observation is that if $S_A \cap S_B =\emptyset$ (a {\em yes}-instance for the set-disjointness problem), then the diameter of $G$ is 2; otherwise  (a {\em {no}}-instance for the set-disjointness problem) the diameter of $G$ is 3. This can be seen as follows. First of all, observe that we must have $\dist(s,t) \leq 2$ unless $s \in V_A$ and $t \in V_B$. Now suppose $s = u_i \in V_A$  and $t = v_j \in V_B$.
 \begin{itemize}
 \item Consider the case $a_i \neq b_j$. We show that $\dist(s,t) = 2$.
      Note that there is an index $l \in [\ell]$ such that $a_i$ and $b_j$ differ at the $l$th bit.
     If the $l$th bit of  $a_i$ is 0 and the $l$th bit of  $b_j$ is 1, then $(u_i, x_l, v_j)$ is a length-2 path between $s$ and $t$.
     If the $l$th bit of  $a_i$ is 1 and the $l$th bit of  $b_j$ is 0, then $(u_i, w_l, v_j)$ is a length-2 path between $s$ and $t$.
 \item Consider the case $a_i = b_j$.   We show that $\dist(s,t) = 3$.
      Note that there is no index $l \in [\ell]$ such that $a_i$ and $b_j$ differ at the $l$th bit.
      Thus, each  $w_l \in V_C$  and $x_l \in V_D$ is adjacent to exactly one of $\{u_i, v_j\}$.
      Hence there is no length-2 path between $s$ and $t$.
\end{itemize}
Therefore, if $S_A \cap S_B =\emptyset$, then $\dist(s,t) = 2$ for all pairs $\{s,t\}$, and so the diameter is 2;
otherwise, there exist $s = u_i \in V_A$  and $t = v_j \in V_B$ such that $\dist(s,t) = 3$, and so the diameter is 3.

\paragraph{Reduction.} Suppose that there is a randomized distributed algorithm $\mathcal{A}$ that is able to compute the diameter with $o(n / \log^2 n)$ energy per vertex, with failure probability $f = 1/\poly(n)$. We show that the algorithm $\mathcal{A}$ can be transformed into a randomized communication protocol that solves the set-disjointness problem with $o(n)$ bits of communication, and with the same failure probability $f = 1/\poly(n)$.

The main challenge in the reduction is that we do not impose any message size constraint.
To deal with this issue, our strategy is to consider a modified computation model $\mathcal{M}'$.
We will endow the vertices in the modified computation model $\mathcal{M}'$ with strictly more capabilities than the original radio network.
Then, we argue that in the setting of $\mathcal{M}'$, we can assume that each message has size $O(\log k)$.

\paragraph{Modified Computation Model $\mathcal{M}'$.} We add the following extra powers to the vertices:
\begin{description}
\item[(P1)] All vertices have access to an infinite shared random string. They know the vertex set and the IDs of all vertices.  Specifically,
$\ID(w_i) = i$ for each $w_i \in V_C$;  $\ID(x_i) = \ell + i$ for each $x_i \in V_D$; $\ID(u^\star) = 2\ell +1$; $\ID(v^\star) = 2\ell +2$. Thus, for each $v \in V_C \cup V_D \cup \{u^\star, v^\star\}$, its role can be inferred from $\ID(v)$.
\item[(P2)] 
Messages received by vertices 
in $V_C \cup V_D \cup \{u^\star, v^\star\}$ (according to the usual radio network rules)
are immediately communicated to \emph{all} vertices. For example, if $v \in V_C \cup V_D \cup \{u^\star, v^\star\}$ receives $m$ from $u \in V$ at time $t$, then by the end of round $t$ all vertices in $V$ know that ``$v$ receives $m$ from $u$ at time $t$.''
\item[(P3)] Each vertex $v \in V_A \cup V_B$ knows the list of the IDs of its neighbors initially.
\end{description}
Next, we discuss the consequences of these extra powers. In particular, we show that we can make the following assumptions about algorithms in this modified model $\mathcal{M}'$.

\paragraph{Vertices in  $V_C \cup V_D \cup \{u^\star, v^\star\}$  Never Transmit.}
Powers (P1) and (P2) together imply that each vertex in the graph is able to locally simulate the actions of all vertices in $V_C \cup V_D \cup \{u^\star, v^\star\}$. Intuitively, this means that all vertices in  $V_C \cup V_D \cup \{u^\star, v^\star\}$ do not need to transmit at all throughout the algorithm.

Note that each vertex $v \in V$ already knows the list of $N(v) \cap (V_C \cup V_D \cup \{u^\star, v^\star\})$. If $v \in V_A \cup V_B$, then $v$ knows this information via (P3). If $v \in V_C \cup V_D \cup \{u^\star, v^\star\}$, then $v$ knows this information via (P1); the role of each vertex in $V_C \cup V_D \cup \{u^\star, v^\star\}$ can be inferred from its ID, which is a public to everyone.

Thus, right before the beginning of each time $t$, each vertex $v \in V$ already knows exactly which vertices in $N(v) \cap (V_C \cup V_D \cup \{u^\star, v^\star\})$  will transmit at time $t$ and their messages. Thus, in the modified model $\mathcal{M}'$, we can simulate the execution of an algorithm which allows the vertices in $V_C \cup V_D \cup \{u^\star, v^\star\}$ to transmit by another algorithm that forbid them to do so.

\paragraph{Messages Sent by Vertices in $V_A \cup V_B$ Have Length $O(\log k)$.}
Next, we argue that we can assume that each message $m$ sent by a vertex $v' \in V_A \cup V_B$ can be replaced by another message $m'$ which contains only the list of all neighbors of $v'$, and this can be encoded as an $O(\log k)$-bit message, as follows. Recall that $N(v')$ is a subset of $V_C \cup V_D \cup \{u^\star, v^\star\}$, and so we can encode $N(v')$ as a binary string of length $|V_C \cup V_D \cup \{u^\star, v^\star\}| = 2\ell + 2 = O(\log k)$.

The message $m$ is a function of all information that $v'$ has. Since no vertex in $V_C \cup V_D \cup \{u^\star, v^\star\}$  transmits any message, $v'$ never receives a message, and so the information that $v'$ has consists of only the following components.
\begin{itemize}
\item The shared randomness and the ID list of all vertices (due to (P1)).
\item The history of vertices in $V_C \cup V_D \cup \{u^\star, v^\star\}$ (due to (P2)).
\item The list of neighbors of $v'$ (due to (P3)).
\end{itemize}
The only private information that $v'$ has is its list of neighbors. If a vertex $u' \in V$ knows the list of neighbors of $v'$, then $u'$ is able to calculate $m$ locally, and so  $v'$ can just send its list of neighbors in lieu of $m$.

\paragraph{Algorithm $\mathcal{A}'$.}
To sum up, given the algorithm $\mathcal{A}$, we can transform it into another algorithm $\mathcal{A}'$ in the modified computation model  $\mathcal{M}'$ that uses only $O(\log k)$-bit messages, and  $\mathcal{A}'$ achieves what $\mathcal{A}$ does. Note that the energy cost of $\mathcal{A}'$ is at most the energy cost $\mathcal{A}$.

\paragraph{Solving Set-Disjointness.} Now we show how to transform  $\mathcal{A}'$ into a protocol for the set-disjointness problem using only $o(k)$ bits of communication.
The protocol is simply a simulation of $\mathcal{A}'$.
The shared random string used by $\mathcal{A}'$ 
is the same random string shared by the two players $A$ and $B$.

Each player $X \in \{A,B\}$ is responsible for simulating vertices in $V_X$.
Vertices in $V_A$ and $V_B$ never receive messages, and so all we need to do is let both players 
$A$ and $B$ know
the messages \emph{sent} to $V_C \cup V_D \cup \{u^\star, v^\star\}$ (in view of (P2)).

We show how to simulate one round  $\tau$ of $\mathcal{A}'$.
Let $Z(\tau)$ be the subset of vertices in $V_C \cup V_D \cup \{u^\star, v^\star\}$
that listen at time $\tau$,
and consider a vertex $u' \in Z(\tau)$.
(Recall that everyone can predict the action of every vertex in $V_C \cup V_D \cup \{u^\star, v^\star\}$.)

Let $Q_A$ be the number of vertices in $N(u') \cap V_A$ transmitting at time $\tau$.
We define $m_{u',\tau,A}$ as follows.
\[
m_{u',\tau,A} =
\begin{cases}
\text{``0''} &\text{if $Q_A = 0$.}\\
\text{``$\geq 2$''} &\text{if $Q_A \geq 2$.}\\
(v', m') &\text{if $Q_A = 1$, and  $v'$ is the vertex in $N(u') \cap V_A$  sending $m'$ at time $\tau$.}
\end{cases}
\]
We define $m_{u',\tau,B}$ analogously. Note that the length of $m'$ must be $O(\log k)$ bits.

The protocol for simulating round $\tau$ is simply that $A$ sends $m_{u',\tau,A}$ (for each  $u' \in Z(\tau)$) to $B$, and $B$ sends $m_{u',\tau,B}$ (for each  $u' \in Z(\tau)$) to $A$.
This offers enough information for both player to know the channel feedback (noise, silence, or a message $m$) received by each vertex in $Z(\tau)$.
Note that the number of bits exchanged by $A$ and $B$ due to the simulation of round $\tau$ is $O(|Z(\tau)|\log k)$.

Recall that the energy cost of each vertex in an execution of $\mathcal{A}'$ is $o(k / \log^2 k)$, and we have $|V_C \cup V_D \cup \{u^\star, v^\star\}| = O(\log k)$.
Thus, the total number of bits exchanged by the two players $A$ and $B$ is
$$\sum_{\tau} O(|Z(\tau)|\log k) = |V_C \cup V_D \cup \{u^\star, v^\star\}| \cdot o(k / \log^2 k)  \cdot O(\log k) = o(k).\qedhere$$
\end{proof}

%\seth{In the sum above, why is there a $\log^2 k$ factor in $|Z(\tau)|\log^2 k$?  Why not $\log k$?}

We remark that the proof of Theorem~\ref{thm:diameter-lb2} can be extended to graphs with higher diameter by using a slightly more complicated lower bound graph construction and analysis; see e.g.,~\cite{bringmann2018note}. Intuitively, this is due to the fact that the lower bound graph is {\em sparse}, so we are able to subdivide the edges. 

\subsection{Upper Bounds}
The approximation ratios in Theorems~\ref{thm:diameter-lb1} and~\ref{thm:diameter-lb2} cannot be improved.
Observe that $\mathsf{BFS}$ already gives a 2-approximation of diameter, as $D' = \max_{u \in V(G)}\{\dist_G(s,u)\} \in [\diam(G)/2, \diam(G)]$,
and we know that a $\mathsf{BFS}$ can be computed in $n^{o(1)}$ energy. 

If we allow an energy budget of $n^{\frac{1}{2} + o(1)}$ then it is possible to achieve a {\em nearly} $3/2$-approximation by applying the algorithm of~\cite{holzer2014brief,RodittyW13},
which computes a $D'$
such that
$\lfloor 2\diam(G)/3 \rfloor  \leq   D'   \leq  \diam(G)$.
More precisely, if we write $\diam(G) =  3h + z$, where $h$ is a non-negative integer, and $z \in \{0,1,2\}$, then $D' \in [2h+z, \diam(G)]$ for the case $z = 0, 1$, and $D' \in [2h+1, \diam(G)]$ 
for the case $z = 2$. Note that this does not contradict the $\Omega(n)$ energy lower bound for distinguishing between $\diam(G)=1$ and $\diam(G)=2$ in Theorem~\ref{thm:diameter-lb1}, nor does it contradict Theorem~\ref{thm:diameter-lb2}.

The algorithm of~\cite{holzer2014brief,RodittyW13} is as follows. Let each vertex join $S$ with probability $(\log n) /\sqrt{n}$, and compute a $\mathsf{BFS}$ from each vertex in $S$. Let $v^\star$ be any vertex that maximizes the distance to $S$. Identify any set of $\sqrt{n}$ vertices $R$ that are the closest to $v^\star$, and compute a $\mathsf{BFS}$ from each vertex in $R$.
The diameter approximation $D'$ is the maximum $\mathsf{BFS}$-label computed throughout the algorithm. Note that there are multiple valid choice of $v^\star$ and $R$, and the tie can be broken arbitrarily.\footnote{Precisely, it is required that $|R| = \sqrt{n}$, and for each $u \in R$, there are less than $\sqrt{n}$ vertices $v$ such that $\dist(v, v^\star) < \dist(u, v^\star)$. In general, there could be multiple choices of $R$ satisfying this requirement.}
Since $\mathsf{BFS}$ can be computed in $n^{o(1)}$ energy, with a suitable implementation, this algorithm be executed using $n^{\frac{1}{2} + o(1)}$ energy.
 For the sake of completeness, in what follows we provide the detail for an implementation, which is based on the following subroutines.

\begin{description}
\item[{\sf Leader Election:}] Elect a leader $v_0\in V$ such that all vertices know $\ID(v_0)$. It is known that this task can be solved in $\tilde{O}(n)$ time and $\tilde{O}(1)$ energy~\cite{ChangDHHLP18}.
\item[{\sf Find Minimum:}] Suppose there is already a leader  $v_0 \in V$, and each vertex $u \in V$ knows $\dist(u,v^\star)$. Each vertex $u$ holds an integer $k_u \in [1, K]$ and a message $m_u$. The goal is to elect one vertex $u^\star$ such that $k_{u^\star} = \min \{ k_u \ | \ u \in V \}$ and have all vertices know  $m_{u^\star}$. Tie is broken arbitrarily. The task {\sf Find Maximum} is defined analogously.
\end{description}

We argue that the task $\findmin$ and $\findmax$ can be solved in $\tilde{O}(\diam(G))$ time and 
$\tilde{O}(1)$ energy, 
given that $K = O(\poly (n))$.
To solve this task, we will do a binary search. Let $I \subseteq [1, K]$ be an interval currently under consideration. We let $v_0$ test whether there exists a vertex $u'$ with $k_{u'} \in I$ by doing $O(\diam(G))$ $\sr$s on the $\mathsf{BFS}$ tree, layer by layer.  The root  $v_0$ is able to announce the result to everyone, also using  $O(\diam(G))$ $\sr$s on the $\mathsf{BFS}$ tree, layer by layer.  
After $O(\log K) = \tilde{O}(1)$ iterations, we are done.

\begin{theorem}\label{thm-diameter-ub1}
 There is an algorithm that computes a 2-approximation  of   diameter   in $n^{1 + o(1)}$ time and $n^{o(1)}$ energy.
\end{theorem}
\begin{proof}
Apply {\sf Leader Election} to elect a leader $v_0$,  do a $\mathsf{BFS}$ from $v_0$, and then do a $\findmax$ to let each vertex learn $\max\{ \dist(u, v_0) \ | \ u \in V\}$. This gives a 2-approximation of the diameter $D$.
\end{proof}

\begin{theorem}\label{thm-diameter-ub2}
 There is an algorithm that computes an approximation $D'$ such that 
 $\lfloor 2\diam(G)/3 \rfloor  \leq   D'   \leq  \diam(G)$ in $n^{3/2 + o(1)}$ time 
 and $n^{1/2 + o(1)}$ energy.
\end{theorem}
\begin{proof}
We show how to implement the algorithm of~\cite{holzer2014brief,RodittyW13}. 
We first apply {\sf Leader Election} to elect a leader $v_0$, and do a $\mathsf{BFS}$ from $v_0$, we will use this tree to do $\findmin$ and $\findmax$ in subsequent steps of the algorithm.

In the algorithm of~\cite{holzer2014brief,RodittyW13}, we let each vertex join $S$ with probability $(\log n) /\sqrt{n}$. Using $|S| = \tilde{O}(\sqrt{n})$ iterations of  $\findmin$  we can let everyone know the $\ID$s of vertices in $S$.
Then, we sequentially compute 
 a $\mathsf{BFS}$ from each vertex in $S$. Let $v^\star$ be a vertex that maximizes the distance to $S$. Such a vertex  $v^\star$ can be elected using one iteration of $\findmax$.
 %We then need to compute a BFS from $v^\star$ and any $\sqrt{n}$ vertices $R$ that are the closest to $v^\star$. 
 To compute the set $R$, we first do a $\mathsf{BFS}$ from $v^\star$ so that everyone knows its distance to  $v^\star$. Then, after $\sqrt{n}$  iterations of $\findmin$, we can let everyone learn the set $R$, and then we can do the $\mathsf{BFS}$ computation from each vertex in $R$ sequentially.
The diameter approximation $D'$ is the maximum $\mathsf{BFS}$-label computed throughout the algorithm, and this can be computed using one iteration of $\findmax$.
It is clear that the algorithm takes $n^{3/2 + o(1)}$ time and $n^{1/2 + o(1)}$ energy, as it only uses $\tilde{O}(\sqrt{n})$ $\findmin$, $\findmax$, and $\mathsf{BFS}$ computations.
\end{proof}

\newpage
\bibliographystyle{abbrv}
\bibliography{reference}

\begin{thebibliography}{10}

\bibitem{AbboudCK16}
A.~Abboud, K.~Censor-Hillel, and S.~Khoury.
\newblock Near-linear lower bounds for distributed distance computations, even
  in sparse networks.
\newblock In C.~Gavoille and D.~Ilcinkas, editors, {\em Distributed Computing
  ({DISC})}, pages 29--42. Springer Berlin Heidelberg, 2016.

\bibitem{alon1991lower}
N.~Alon, A.~Bar-Noy, N.~Linial, and D.~Peleg.
\newblock A lower bound for radio broadcast.
\newblock {\em Journal of Computer and System Sciences}, 43(2):290--298, 1991.

\bibitem{bar1991efficient}
R.~Bar-Yehuda, O.~Goldreich, and A.~Itai.
\newblock Efficient emulation of single-hop radio network with collision
  detection on multi-hop radio network with no collision detection.
\newblock {\em Distributed Computing}, 5(2):67--71, 1991.

\bibitem{bar1992time}
R.~Bar-Yehuda, O.~Goldreich, and A.~Itai.
\newblock On the time-complexity of broadcast in multi-hop radio networks: An
  exponential gap between determinism and randomization.
\newblock {\em Journal of Computer and System Sciences}, 45(1):104--126, 1992.

\bibitem{BarnesCMA10}
M.~Barnes, C.~Conway, J.~Mathews, and D.~K. Arvind.
\newblock {ENS}: An energy harvesting wireless sensor network platform.
\newblock In {\em Proceedings of the 5th International Conference on Systems
  and Networks Communications ({ICSNC})}, pages 83--87, 2010.

\bibitem{BenderKPY18}
M.~Bender, T.~Kopelowitz, S.~Pettie, and M.~Young.
\newblock Contention resolution with constant throughput and log-logstar
  channel accesses.
\newblock {\em SIAM J. Comput.}, 47:1735--1754, 2018.

\bibitem{BerenbrinkCH09}
P.~Berenbrink, C.~Cooper, and Z.~Hu.
\newblock Energy efficient randomised communication in unknown adhoc networks.
\newblock {\em Theoretical Computer Science}, 410(27):2549 -- 2561, 2009.

\bibitem{BravermanM13}
M.~Braverman and A.~Moitra.
\newblock An information complexity approach to extended formulations.
\newblock In {\em Proceedings of the 45th Annual ACM Symposium on Theory of
  Computing ({STOC})}, pages 161--170, New York, NY, USA, 2013. ACM.

\bibitem{bringmann2018note}
K.~Bringmann and S.~Krinninger.
\newblock A note on hardness of diameter approximation.
\newblock {\em Information Processing Letters}, 133:10--15, 2018.

\bibitem{ChangDHHLP18}
Y.-J. Chang, V.~Dani, T.~P. Hayes, Q.~He, W.~Li, and S.~Pettie.
\newblock The energy complexity of broadcast.
\newblock In {\em Proceedings of the 2018 {ACM} Symposium on Principles of
  Distributed Computing ({PODC})}, pages 95--104, 2018.

\bibitem{ChangKPWZ17}
Y.-J. Chang, T.~Kopelowitz, S.~Pettie, R.~Wang, and W.~Zhan.
\newblock Exponential separations in the energy complexity of leader election.
\newblock In {\em Proceedings of the 49th Annual {ACM} {SIGACT} Symposium on
  Theory of Computing ({STOC})}, pages 771--783, 2017.

\bibitem{chlamtac1985broadcasting}
I.~Chlamtac and S.~Kutten.
\newblock On broadcasting in radio networks-problem analysis and protocol
  design.
\newblock {\em IEEE Transactions on Communications}, 33(12):1240--1246, 1985.

\bibitem{ChlamtacK87}
I.~Chlamtac and S.~Kutten.
\newblock Tree-based broadcasting in multihop radio networks.
\newblock {\em {IEEE} Trans. Computers}, 36(10):1209--1223, 1987.

\bibitem{CzumajD17}
A.~Czumaj and P.~Davies.
\newblock Exploiting spontaneous transmissions for broadcasting and leader
  election in radio networks.
\newblock In {\em Proceedings of the 2017 {ACM} Symposium on Principles of
  Distributed Computing ({PODC})}, pages 3--12, 2017.

\bibitem{GasieniecKKPS07}
L.~Gasieniec, E.~Kantor, D.~R. Kowalski, D.~Peleg, and C.~Su.
\newblock Energy and time efficient broadcasting in known topology radio
  networks.
\newblock In {\em Proceedings 21st International Symposium on Distributed
  Computing ({DISC})}, pages 253--267, 2007.

\bibitem{GhaffariH16}
M.~Ghaffari and B.~Haeupler.
\newblock Near-optimal {BFS}-tree construction in radio networks.
\newblock {\em {IEEE} Communications Letters}, 20(6):1172--1174, 2016.

\bibitem{GilbertKPPSY14}
S.~Gilbert, V.~King, S.~Pettie, E.~Porat, J.~Saia, and M.~Young.
\newblock ({N}ear) optimal resource-competitive broadcast with jamming.
\newblock In {\em Proceedings of the 26th {ACM} Symposium on Parallelism in
  Algorithms and Architectures ({SPAA})}, pages 257--266, 2014.

\bibitem{haeupler2016faster}
B.~Haeupler and D.~Wajc.
\newblock A faster distributed radio broadcast primitive.
\newblock In {\em Proceedings 35th ACM Symposium on Principles of Distributed
  Computing ({PODC})}, pages 361--370. ACM, 2016.

\bibitem{holzer2014brief}
S.~Holzer, D.~Peleg, L.~Roditty, and R.~Wattenhofer.
\newblock Brief announcement: Distributed 3/2-approximation of the diameter.
\newblock In {\em Proc. 28th International Symposium on Distributed Computing
  (DISC 2014)}, pages 562--564. Springer, 2014.

\bibitem{JurdzinskiKZ02}
T.~Jurdzinski, M.~Kutylowski, and J.~Zatopianski.
\newblock Efficient algorithms for leader election in radio networks.
\newblock In {\em Proceedings of the 21st Annual {ACM} Symposium on Principles
  of Distributed Computing ({PODC})}, pages 51--57, 2002.

\bibitem{JurdzinskiKZ02b}
T.~Jurdzinski, M.~Kutylowski, and J.~Zatopianski.
\newblock Energy-efficient size approximation of radio networks with no
  collision detection.
\newblock In {\em Proceedings of the 8th Annual International Conference on
  Computing and Combinatorics ({COCOON})}, pages 279--289, 2002.

\bibitem{JurdzinskiKZ02c}
T.~Jurdzinski, M.~Kutylowski, and J.~Zatopianski.
\newblock Weak communication in radio networks.
\newblock In {\em Proceedings of the 8th International European Conference on
  Parallel Computing ({Euro-Par})}, pages 965--972, 2002.

\bibitem{JurdziskiKZ03}
T.~Jurdzinski, M.~Kutylowski, and J.~Zatopianski.
\newblock Weak communication in single-hop radio networks: adjusting algorithms
  to industrial standards.
\newblock {\em Concurrency and Computation: Practice and Experience},
  15(11--12):1117--1131, 2003.

\bibitem{JurdzinskiS02}
T.~Jurdzinski and G.~Stachowiak.
\newblock Probabilistic algorithms for the wakeup problem in single-hop radio
  networks.
\newblock In {\em Proceedings of the 13th International Symposium on Algorithms
  and Computation ({ISAAC})}, pages 535--549, 2002.

\bibitem{KabarowskiKR06}
J.~Kabarowski, M.~Kutylowski, and W.~Rutkowski.
\newblock Adversary immune size approximation of single-hop radio networks.
\newblock In {\em Proceedings Third International Conference on Theory and
  Applications of Models of Computation ({TAMC})}, pages 148--158, 2006.

\bibitem{KalyanasundaramS92}
B.~Kalyanasundaram and G.~Schnitger.
\newblock The probabilistic communication complexity of set intersection.
\newblock {\em {SIAM} J. Discrete Math.}, 5(4):545--557, 1992.

\bibitem{KardasKP13}
M.~Kardas, M.~Klonowski, and D.~Pajak.
\newblock Energy-efficient leader election protocols for single-hop radio
  networks.
\newblock In {\em Proceedings 42nd International Conference on Parallel
  Processing ({ICPP})}, pages 399--408, 2013.

\bibitem{KingPSY18}
V.~King, S.~Pettie, J.~Saia, and M.~Young.
\newblock A resource-competitive jamming defense.
\newblock {\em Distributed Computing}, 31:419--439, 2018.

\bibitem{KlonowskiP18}
M.~Klonowski and D.~Pajak.
\newblock Brief announcement: Broadcast in radio networks, time vs. energy
  tradeoffs.
\newblock In {\em Proceedings 37th {ACM} Symposium on Principles of Distributed
  Computing ({PODC})}, pages 115--117, 2018.

\bibitem{KlonowskiS16}
M.~Klonowski and M.~Sulkowska.
\newblock Energy-optimal algorithms for computing aggregative functions in
  random networks.
\newblock {\em Discrete Mathematics {\&} Theoretical Computer Science},
  17(3):285--306, 2016.

\bibitem{KushilevitzM98}
E.~Kushilevitz and Y.~Mansour.
\newblock An {$\Omega(D\log (N/D))$} lower bound for broadcast in radio
  networks.
\newblock {\em SIAM Journal on Computing}, 27(3):702--712, 1998.

\bibitem{KutylowskiR03}
M.~Kutylowski and W.~Rutkowski.
\newblock Adversary immune leader election in ad hoc radio networks.
\newblock In {\em Proceedings 11th Annual European Symposium on Algorithms
  ({ESA})}, pages 397--408, 2003.

\bibitem{MillerPVX15}
G.~L. Miller, R.~Peng, A.~Vladu, and S.~C. Xu.
\newblock Improved parallel algorithms for spanners and hopsets.
\newblock In {\em Proceedings of the 27th {ACM} on Symposium on Parallelism in
  Algorithms and Architectures ({SPAA})}, pages 192--201, 2015.

\bibitem{miller2013parallel}
G.~L. Miller, R.~Peng, and S.~C. Xu.
\newblock Parallel graph decompositions using random shifts.
\newblock In {\em Proceedings of the 25th Annual ACM Symposium on Parallelism
  in Algorithms and Architectures ({SPAA})}, pages 196--203, 2013.

\bibitem{NakanoO00}
K.~Nakano and S.~Olariu.
\newblock Energy-efficient initialization protocols for single-hop radio
  networks with no collision detection.
\newblock {\em IEEE Trans. Parallel Distrib. Syst.}, 11(8):851--863, 2000.

\bibitem{Newport14}
C.~Newport.
\newblock Radio network lower bounds made easy.
\newblock In {\em Proceedings of the 28th International Symposium on
  Distributed Computing ({DISC})}, pages 258--272, 2014.

\bibitem{PolastreSC05}
J.~Polastre, R.~Szewczyk, and D.~Culler.
\newblock Telos: enabling ultra-low power wireless research.
\newblock In {\em Proceedings of the 4th International Symposium on Information
  Processing in Sensor Networks ({IPSN})}, pages 364--369, 2005.

\bibitem{RodittyW13}
L.~Roditty and V.~V. Williams.
\newblock Fast approximation algorithms for the diameter and radius of sparse
  graphs.
\newblock In {\em Proceedings 45th ACM Symposium on Theory of Computing
  (STOC)}, pages 515--524, 2013.

\end{thebibliography}

%\appendix

\end{document}